\documentclass[12pt]{article}
\usepackage{amsfonts}
\usepackage{amssymb}
\usepackage{graphicx}
\usepackage{amsmath}
\usepackage{subcaption,tikz,pgfplots}
\usepackage{hyperref}

\usepackage{amstext}
\usepackage{geometry}
\usepackage[onehalfspacing]{setspace}
\usepackage{amsmath}
\DeclareMathOperator*{\argmax}{arg\,max}

\newtheorem{theorem}{Theorem}

\newtheorem{lemma}{Lemma}

\newtheorem{proposition}{Proposition}

\newenvironment{proof}[1][Proof]
{\textbf{#1.} }{\ \rule{0.5em}{0.5em}}

\usepackage[longnamesfirst]{natbib}
\begin{document}

\title{Gacha Game: \\ When Prospect Theory Meets Optimal Pricing}
\author{ Tan Gan\thanks{I thank Nicholas Barberis, Dirk Bergemann, Krishna Dasaratha, Johannes H\"{o}rner, Ryota Iijima, Jonathan Libgober, Weicheng Min, Philipp Strack, Kaihao Yang, and Xingtan Zhang for helpful suggestions. I would also like to thank conference participants at the 33rd Stony Brook conference, Econometric society meeting (China 2021 summer, NA 2022 summer), 92nd SEA conference, and 2022 YES conference for their helpful comments. The current draft has been edited by AJE.
		Tan Gan: Department of Economics, Yale University, New Haven, CT 06511, \texttt{%
			tan.gan@yale.edu}.}}

\date{\today }
\maketitle

\begin{abstract}
I study the optimal pricing process for selling a unit good to a buyer with prospect theory preferences. In the presence of probability weighting, the buyer is dynamically inconsistent and can be either sophisticated or naive about her own inconsistency.

If  the buyer is naive, the uniquely optimal mechanism  is to sell  a ``loot box'' that delivers the good with some constant probability in each period. 
In contrast, if the buyer is sophisticated, the uniquely optimal mechanism introduces worst-case insurance: after  successive failures in obtaining the good from all previous loot boxes, the buyer can purchase the good  at full price.

\end{abstract}


\medskip

\newpage

\section{Introduction}

In various entertainment-related sectors, including card collection, penny auctions, crane machines, blind boxes, and video games, products are often sold using stochastic pricing processes rather than fixed posted prices. A prevailing justification for this approach is that consumers exhibit gambling preferences that diverge from expected utility maximization. Among the decision-theoretic models of risk attitudes, \cite{TVERSKYKAHNEMAN79}'s (cumulative) prospect theory is one of the most prominent alternatives that explains preference for gambling. This paper explore the pricing problem for an expected-revenue-maximizing seller who attempts to sell a unit good to a cumulative prospect theory (CPT) buyer, with a particular focus on probability weighting.\footnote{Other elements of prospect theory, like loss aversion, are not the primary focus of this paper. For instance, Section 4.1 demonstrates that my principal finding remains valid for rank-dependent expected utility, where only probability weighting is featured.} The optimal dynamic design provides a theoretical rational for pricing processes observed in reality.

As the leading example of this paper, gacha games are video games that employ loot boxes as their central monetization strategy. Selling loot boxes is becoming increasingly popular in the gaming industry, whose market size in the US reached \$85 billion  in 2021.\footnote{In comparison, video streaming apps (including Youtube, Netflix and Tiktok) generated 24.1 billion in the US in 2020. All North American sports had a market size of \$71 billion in 2018 (pre-pandemic).  The movie industry generated \$11.89 billion in box office revenue  in the US \& Canada in 2018 (pre-pandemic).  The music industry had a market size of \$9.8 billion in the US in 2021. These numbers come from \url{https://www.statista.com}.} In fact, 60\% of the leading games on both the Google Play store and iPhone App store incorporated loot boxes, as highlighted by \cite{Zendle2020}.
Consider the case of Genshin Impact. The gacha game achieved \$3 billion in revenue during its first year. To obtain valuable virtual weapons in this game, players must purchase multiple loot boxes. Priced at \$2, each loot box offers a 0.6\% chance of providing the desired product. If players do not get the product even after 90 purchases, the game ensures they receive it—a situation that arises with a 58\% likelihood. This mechanism, termed the "pity system," originated with Granblue Fantasy in 2016 and soon became a feature in many gacha games, albeit with variations in parameters.
In contrast, prior to the innovation of Granblue Fantasy,  the pricing processes in gacha games were stationary. For example, a stationary loot design, which leads to the same expected payment as the example before, would price each loot box at \$2 and delivers the product with probability 1.6\%.

These two different mechanisms in reality match the  theoretically optimal dynamic pricing process for sophisticated and naive buyers, respectively. This distinction between sophistication and naivety comes from the dynamic inconsistency caused by probability weighting, as noted by \cite{Barberis12}. Manipulating the dynamically inconsistent incentive constraints caused by probability weighting is also the key feature and challenge of the model.

Apart from probability weighting, the buyer in the model has a standard quasilinear utility function, and I assume that her valuation of the unit good is publicly known in most parts of this paper. In this simple setup, the pricing problem for the standard EUT buyer is trivial: the seller simply sells the good deterministically at a price equal to the buyer's valuation. Focusing on this simple setting highlights that introducing randomization has nothing to do with  screening the private types of the buyer.

I first present the  analysis with a sophisticated buyer. The main result considers the design problem of a continuous-time dynamic pricing process. The seller designs a deterministic cumulative payment function and determines a random stopping time when the good is delivered to the buyer. In the direct mechanism, the buyer always remains in the process and increases her cumulative payment over time according to the payment function  until the good is delivered according to the random stopping time. 
The incentive constraint for the design problem is dynamic individual rationality (DIR), which requires that, at each time before the buyer receives the good, she finds that remaining in the process is preferable to quitting and forgoing previous payment. At each time, the buyer updates her belief about the stopping time, conditional on the event that the good has not been delivered. Bayesian updating, in conjunction with nonlinear probability weighting, causes dynamic inconsistency. The designer must track all posteriors, instead of posterior means, to ensure that the DIR constraints are satisfied, which is the main technical challenge of the framework.

The optimal mechanism is unique in terms of the distribution of the total payment and can be approximated by the following discrete-time dynamic process. At the beginning of the process, the seller sells loot boxes at a fixed price, each of which delivers the good with a fixed probability. If the buyer is unlucky and does not obtain the good after purchasing $N$ boxes, the seller will sell the good deterministically to her at a price that equals her value. This pricing process resembles the current pricing process adopted in the gacha games industry in that it modifies the stationary loot box mechanism by providing worst-case insurance at the end of the process. The details of the insurance mechanisms at the end of the two processes appear different, but with optimally tuned parameters, the difference in profit is very small. This is because both mechanisms protect the buyer from the large tail payment, which the CPT buyer overevaluates, and attach an intermediate tail probability, which the CPT buyer underevaluates, to the majority of its payment realizations.\footnote{See Section 3.5 for detailed discussion. For readers who believe in market efficiency and entrepreneurship, this similar performance suggests the plausibility of an optimal mechanism. For others, the small difference raises an interesting empirical question of whether the  optimal mechanism can lead to higher revenue in real games, as the probability design of gacha systems is an ongoing innovation in the industry.} 

Stationary pricing processes perform poorly when the buyer is sophisticated. In contrast, they are optimal when the buyer is naive about her dynamic inconsistency.  I consider a discrete-time dynamic design problem. In each period, the seller can provide one  menu that consists of an allocation probability and a random price, and the buyer decides whether to purchase. In the generically unique optimal solution, the seller sells a ``loot box'' that delivers the good with some constant probability in each period.  The naive buyer always believes that she will try her luck just one last time and quit thereafter. However, she ultimately continues purchasing in every period until she finally receives the good.

From a theoretical perspective, this paper is the first study to characterize a dynamic design problem where the agent—whether sophisticated or naive—exhibits dynamic inconsistency due to probability weighting. The methods I  develop might inspire future theoretical exercises in this class of problem. From a modeling standpoint, this paper introduced a crucial assumption: payment can only be one-sided. This constraint is suitable for discussing optimal pricing for product sales, as numerous nations have explicitly prohibited gaming companies from remunerating consumers. The typical argument is that setting negative prices is tantamount to gambling, which is either illegal or necessitates specific authorization in many jurisdictions. If the price can be arbitrarily negative, as it is in the gambling industry, then the seller can extract infinite profit from the buyer with reverse-S shape probability weighting (\cite{AzevedoGottlieb12}).

The discernible difference in profit, with and without the non-negativity constraint on price, suggests a theoretical distinction between stochastic pricing processes and gambling. This observation might contribute to the ongoing discussions surrounding the regulation of loot boxes,\footnote{For example, Spain is considering tighter restrictions on gacha games by requiring players to provide official proof of age and implementing mandatory spending limits. The UK government has called for the purchase of loot boxes to be made unavailable to children and young people unless they are approved by a parent or guardian. The Netherlands government is proposing a much harsher rule that would completely ban the use of loot boxes in the country.} by offering a perspective that supports a differentiated regulatory approach. However, a mere prohibition on sellers transferring money to buyers does not adequately ensure that the payment is one-sided. Buyers must also be restricted from profiting through the resale of these products; failing this, sellers may circumvent the constraint by introducing multiple items suitable for resale in secondary markets, essentially realizing a covert negative price. Given this, regulatory authorities should exhibit heightened vigilance when firms employ stochastic pricing, particularly for tangible goods or transferable virtual items that boast considerable secondary market value. This perspective aligns with extant legal paradigms. For instance, a 2018 study by the Netherlands Gaming Authority on ten undisclosed games found that four contravened national gambling regulations. The contention was that loot box prizes from these games held intrinsic market value, given their tradeability outside the game environment.

Beyond the gacha game industry, the theoretical model might shed light on other economic applications. As an illustration, the discussion session lists two applications where the main result characterizes the optimal solution in a equivalent model. The first concerns how  a firm could optimally design the random arrival time of a promotion to maximize the expected effort a CPT worker exerts, subject to his or her dynamic participation constraints. The second application concerns how a queuing system could optimally design the random arrival time of the service to minimize the expected burden to the system, subject to the dynamic constraint to keep the CPT agent in line.

\paragraph{Literature Review} 
This paper relates to the literature that explores the implication of the \cite{KoszegiRabin06,KoszegiRabin07} loss-aversion on pricing or mechanism design problems: \cite{HeidhuesKoszegi08},   \cite{KoszegiRabin09}, \cite{HerwegMuller10}, \cite{LangeRatan10}, \cite{CARBAJALEly13}, \cite{KoszegiHeidhues14}, \cite{ElyCarbajal16} and \cite{Rosato16}. 
My paper is orthogonal or complement to this line of research because the driving force of the paper is the reverse-S shape probability weighting, which leads to a preference for certain forms gambling. In contrast, as \cite{MR16} point out, the linear utility under \cite{KoszegiRabin06} loss-averse framework corresponds to a concave probability weighting function in my specification, which leads to risk aversion. 
For example, \cite{KoszegiHeidhues14} considers the static optimal  pricing problem with a loss-averse agent, who has known value on the product. They show a random price is optimal when the buyer features separate evaluation, which is driven by the buyer's aversion to lose the product that is expected to be purchased. The buyer in their model suffers from randomness of the payment and if the buyer features joint evaluation, a deterministic price is optimal. My paper consider a dynamic pricing problem with a buyer, who has known value and intrinsically seeks certain risk driven by probability weighting. My main result holds under both separate evaluation and joint evaluation, as it is also applicable to rank-dependent expected utility (RDEU) as in  \cite{QUIGGIN1982}.

There are a few papers that explore implications of other non-expected utility framework on design problems. For example,
\cite{GMSZ21} consider the auction design for bidders equipped with non-expected utility preferences that exhibit constant risk aversion (CRA), which in their framework includes \cite{KoszegiRabin06}'s loss aversion.  \cite{GMSZ23} further study a generalization of the classical monopoly insurance problem under adverse selection. In contrast, this paper studies agents who have reverse-S shape probability weighting and exhibit behavior patterns of risk-seeking.
For more detailed reviews on  behavioral mechanism design or behavioral contract theory, see \cite{Spiegler11} and \cite{2014Botond}.

More broadly, this paper contributes to the literature on optimal dynamic design problem with dynamically inconsistent buyers. The current literature has discussed inconsistency caused by temptation \cite{AmadorWerningAngeletos06} or from present bias \cite{RabinDonoghue01}, \cite{HeidhuesKoszegi10}, \cite{MarinaPierre14}, \cite{Galperti15},\cite{Bard20},\cite{GottliebZhang21}.  My paper is the first to explore this problem where the buyer's dynamic inconsistency comes from her belief updating about a future random payment, as highlighted by \cite{Barberis12} and \cite{StrackEbert15}. 
 This is a completely different channel because buyers with probability weighting will have dynamically consistent preferences if the future allocation is deterministic and the key of the design is to manipulate the randomness. There is a related paper in applied probability \cite{XYZ13} that explores the optimal precommitted stopping time for a geometric Brownian motion when the decision maker features probability weighting. Their precommitment precludes the impact of dynamic inconsistency and the random process in their framework is not a design object.

From the perspective of applications, by offering a profit-maximization rationale for the design of loot boxes, this paper contributes to the fast-growing and interdisciplinary literature on loot boxes in video games: \cite{DrummondS18},  \cite{CEHL21},
\cite{XHNN22}, \cite{LEMMENS2022100023}. \cite{CDW23}, or more broadly, the empirical studies on the implications of consumers' preferences for gambling  in certain industries. For example, the retail industry \cite{MSA17}; the  lottery industry, \cite{Clotfelter90}, \cite{KEARNEY2005}, and \cite{Benjamin21}.

Prospect theory is an important framework but not the only one used to empirically explain the preference for risk. Other candidates include RDEU (\cite{QUIGGIN1982}), risk-seeking EU or S-shaped utility, and salience theory \cite{Shleifer12,Shleifer13b}. My main results also hold under RDEU, as discussed in Section 4.1 and 4.3. The design problem is not well-defined under risk-seeking EU or S-shaped utility even with one-sided payment. I discuss the implications of salience theory in Section 6.2, where I show that the basic result remains qualitatively the same.

\section{Model: Sophisticated Buyer}
\paragraph{Environment}
A seller is selling a unit good to a buyer who behaves according to the cumulative prospect theory (CPT). Whenever the buyer is making a decision, she needs to evaluate $(X,T)$, a probability distribution $\Delta (\{0,1\} \times \mathbb{R}^+)$.
$X$ denotes a binary random variable that represents the random allocation of the unit good, where $X=1$ means the buyer gets the product. $T$ denotes the non-negative random price that the buyer has to pay. 
\paragraph{The Buyer's Preference}
The payoff function $V(X,T)$ of the buyer is defined as follows. 
\begin{gather*}
	V(X,T)
	=\theta w_+ (P(X  =1 )  ) - \lambda \int_{0}^{\infty} w_-( P(  T > y )  ) dy,
\end{gather*}
where $\theta$ is the buyer's value of the product and is commonly known. $w_{+}$ and $w_-: [0,1] \to [0,1]$ are two  probability weighting functions. $w_{\pm}(p)$ is strictly increasing with $w_{\pm}(0)=0$, $w_{\pm}(1)=1$. Moreover, as illustrated in Figure \ref{pic6}, the probability function takes a reverse S shape:  $w'_{\pm}(0),w'_{\pm}(1)>1$, and $w'''_{\pm}(p) > 0$.\footnote{This condition implies that there exists $p_{\pm} \in (0,1)$  such that $w_{\pm}$ is strictly concave in $[0,p_{\pm}]$ and strictly convex in $[p_{\pm},1]$, Strictness is not needed for all theorems about optimality. I only need it for uniqueness.}
I illustrate $w(p)$ in Figure \ref{pic6} using the classical parameterized function $\frac{p^{\gamma}}{(p^{\gamma} + (1-p)^{\gamma}  )^{1/\gamma}   }$
proposed by  \cite{TVERSKYKAHNEMAN92}. Throughout the paper, all plots and numerical calculations use this specific functional form, and I maintain $\gamma$ to denote its parameter. 

 Note that in the special case where $w_-(p) = w_+(p)=  p$ and $ \lambda = 1$, the payoff function $V$ reverts to the standard expression for the expected utility theory (EUT).
\begin{align*}
	V_E  &= \theta P(X = 1) - \mathbb{E} T 
	 = \theta P(X = 1) -  \int_{0}^{\infty} P(  T > y )   dy
\end{align*}
To better illustrate the difference between CPT and EUT, suppose that the random payment $T$ yields a density $f$. Then,
\begin{gather*}
	\int_{\mathbb{R}^+} w_- (P(T>y) dy= \int_0^{\infty}    y f(y) w_-'(P(T>y) dy.
\end{gather*}
The density form intuitively illustrates the idea that prospect theory consumers overevaluate the probability of extreme events since $w'_-(p)>1$ if and only if $p$ is close to 0 or 1.

\begin{figure}[h]
		\centering
		\includegraphics[width=8cm]{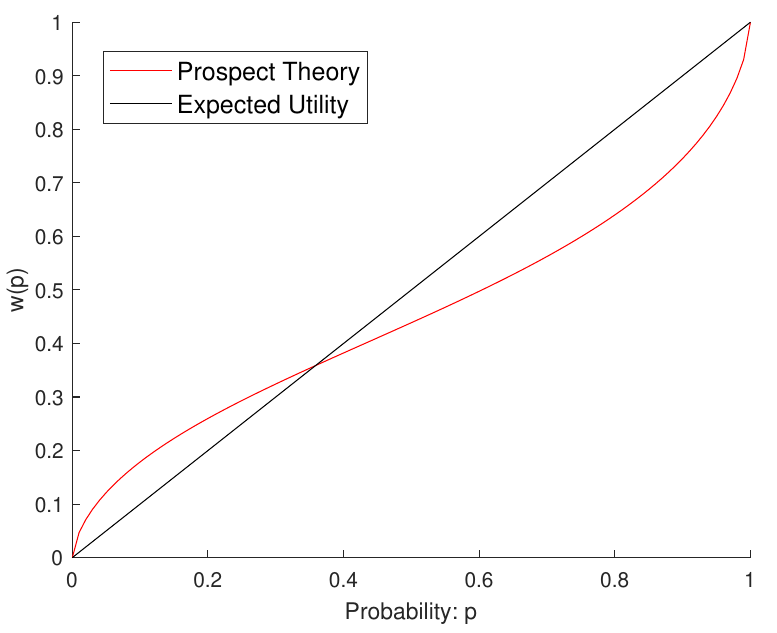}
		\hspace{0.1cm}
	\caption{Probability Weighting Function $
		\frac{p^{\gamma}}{(p^{\gamma} + (1-p)^{\gamma}  )^{1/\gamma}}$ with $\gamma=0.65$.}
			\label{pic6}
\end{figure}

There are several underling model choices behind the formulation of the payoff function $V$. In particular, I assume the buyer is risk neutral, the reference point is the initial condition at the time of decision, and separate evaluation. I elaborate the details  in Section 4.2, and describe the robustness of results to other specifications (joint accounting, rank-dependent expected utility, salience theory) in Sections 4.1, 4.3 and 6.2.

\paragraph{Design Objectives}
The seller designs a continuous time dynamic pricing process. In particular, the seller designs a stopping time $\tau$ over $\mathbb{R}^+ \cup +\infty$ and a deterministic cumulative payment function $T(t) : \mathbb{R}^+ \to \mathbb{R}^+$. The stopping time $\tau$ represents the time when the buyer obtains the good, and the cumulative payment $T(t)$ characterizes the payment the buyer makes if she remains in until (but not including) time $t$. 

At each time $t$, the buyer observes whether  she receives the good $\{\tau \leq t \}$. If not, she needs to decide whether to remain in the process. If she quits, she forgoes her cumulative payment $T(t)= T(t-) = \lim_{s <t \to t} T(s)$, whereas if she remains in the process, she needs to increase her payment according to $T(t)$. If  $T(t+)-T(t) = \lim_{s >t \to t} T(s) - T(t) > 0$ at $t$, she needs to pay a lump sum payment at that instance. 
$T(0) = 0$, and the cumulative payment $T(t)$ is weakly increasing over time.

\paragraph{Discussion about Pricing Processes}

The pricing processes is a generalization of static post price: to sell the product deterministically at a post price $p_0$ one can set $\tau=t_0>0$ and $T(t) = 1_{t>0} p_0$.  The rationale behind constraining the cumulative price 
$p(t)$ to be both deterministic and increasing stems from the fact that, within the relevant context, \emph{payments are one-sided, and consumers lack the commitment to future payments}. Ultimately, all transactions need deterministic execution. In order to implement an ex-ante random payment, the seller, endowed with commitment power, has to design a random allocation rule and a deterministic payment path that ensures it remains optimal for the buyer to complete payment at any given time in the process.\footnote{If $T(t)$ can be either random or decreasing, then the optimal solution coincides with the benchmark characterized in Proposition \ref{prop1}, where the buyer has commitment power. If $T(t)$ can be arbitrarily negative, then the profit is unbounded according to the construction of \cite{AzevedoGottlieb12}.}

\paragraph{No Discounting}
Because the dynamic pricing process is just a way to implement ex-ante random payment while addressing the buyer's lack of commitment, the seller can make the iterative process ``as fast'' as he wants. This
negates the need for discounting within the model. Should discounting be incorporated, the supremum attainable profit would equal the optimal profit characterized in my framework. It can not be exactly  achieved by any pricing processes, but can be approximated by infinitely accelerating the optimal process characterized in the paper.

\paragraph{Dynamic Individual Rationality}
The buyer is sophisticated, so that she can predict what herself does in the mechanism. Thus, I invoke revelation principle and focus on direct mechanisms in which the buyer remains in the process until she obtains the good. The DIR condition requires that at each time $s~ s.t. ~ P(\tau > s) > 0$, the consumer believes that remaining forever is weakly better than quitting immediately. At the time of the decision $s$, the buyer has not got the product $(\tau>s)$, and she will eventually get the product if $\tau<\infty$ and the additional payment is $T(\tau)-T(s)$, so the DIR is
\begin{gather*}
 \theta   w_+(P(\tau < \infty | \tau >s )) - \lambda \int_{0}^{\infty} w_-( P( T (\tau) - T(s) > y  | \tau >s )  ) dy \geq 0.
\end{gather*}

\paragraph{Regularity of $T(t)$}
$T(t)$ is left continuous by definition. $T(t)$ is also increasing because at any moment, the payment has to be nonnegative. In addition, I focus on $T(t)$ with the following ``regular decomposition'' $\mspace{-8.2mu}$.\footnote{According to the Lebesgue decomposition theorem, an increasing function $T(t)$ might have a third additive separable component: $T(t) = T_1(t) + T_2(t) + T_3(t)$. This last component $T_3(t)$ is ``irregular'' in the sense that it is continuous and its derivative is almost everywhere 0, but it is not a constant function. I ignore this third part in the design problem.}
\begin{gather*}
	T(t) = \int_0^t f(t) + \sum_{i=0}^{\infty} 1_{t_i < t} T_i,
\end{gather*}
where $f(t)$ is a nonnegative integrable function, $(t_i)_i$ and $(T_i)_i$ are  countable sequences of nonnegative numbers with $t_i$ increasing in $i$. The interpretation is that the buyer pays a flow payment of $f(t)$ at any time $t$ and pays a lump sum payment $T_i$ at time $t_i$ for $i=1,2,...,\infty$.

In the same spirit, I focus on $\tau$, whose cumulative density function, $F(t) = P(\tau \leq t)$, admits a regular decomposition: 
\begin{gather*}
	F(t) = F_1(t) + F_2(t),
\end{gather*}
where $F_1$ is absolutely continuous and $F_2$ is a jump function. Both $F_1$ and $F_2$ are increasing and nonnegative.

\paragraph{The Design Problem}
The seller chooses the optimal stopping time $\tau$ and cumulative payment $T(t)$ to maximize his expected profit:
\begin{gather}
	\max_{\tau,T(t)}  \int_0^{\infty} P(T(\tau ) > y)  dy,  \nonumber\\
	s.t.  \quad  \theta   w_+(P(\tau < \infty | \tau >s )) - \lambda \int_{0}^{\infty} w_-( P( T (\tau) - T(s) > y  | \tau >s )  ) dy \geq 0, \nonumber \\  \quad \forall s,~ s.t. ~ P(\tau > s) > 0  \qquad \text{(DIR)} \label{eq_IC}.
\end{gather}

\section{Analysis}
\subsection{Benchmark: Ex ante IR}
This subsection considers the design problem where only the ex ante IR of the buyer is respected. Equivalently, this is the design problem where the buyer can commit on the behavior of her future self. The purpose is twofold. First, it illustrates the basic driving force of the reverse-S shape probability weighting function, and helps to understand the shape of the optimal pricing process in the main result. Second, it illustrates that the dynamic process is beneficial only because it can solve the issue of participation caused by the dynamic inconsistency. The benchmark design problem with only ex ante IR is 
\begin{gather*}
	\max_{\tau,T(t)}  \int_0^{\infty} P(T(\tau ) > y)  dy,  \nonumber\\
	s.t.  \quad  \theta   w_+(P(\tau < \infty )) - \lambda \int_{0}^{\infty} w_-( P( T (\tau)  > y  )  ) dy \geq 0.
\end{gather*}

\begin{figure}[h]
	\centering
		\centering
		\includegraphics[width=8cm]{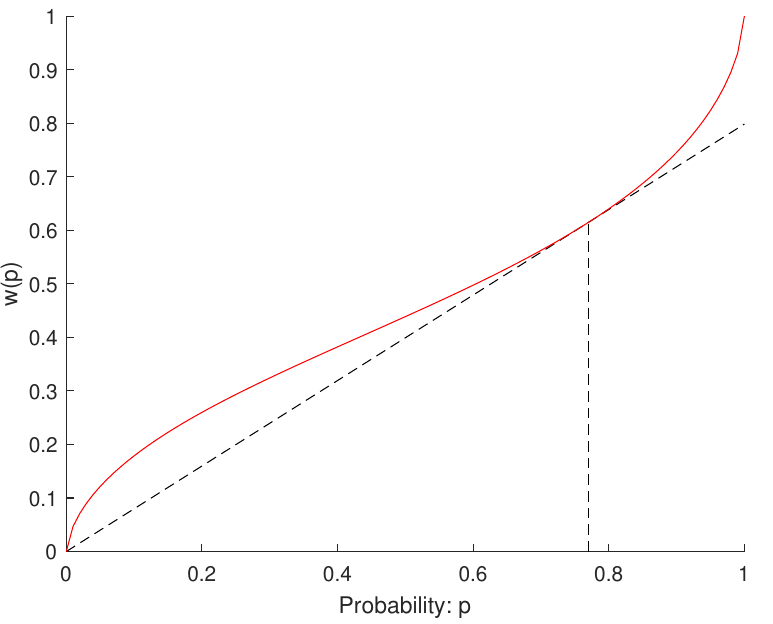}
		\caption{Determining the Multiplier} 	\label{pic2}
		\hspace{0.1cm}
\end{figure}

Clearly, it is optimal to set $P(\tau < \infty )=1$ or, equivalently, to always deliver the good to the buyer. Additionally, designing $T(\tau)$ is equivalent to designing the tail distribution function $P(y) = P(T(\tau)>y)$ such the design problem can be rewritten as:
\begin{gather}
	\label{program1}
	\Pi(\theta) = \max_{P(y)}   \int_0^{\infty} P(y) dy , \\
	s.t. \quad    \lambda \int_0^{+\infty} w_-( P( y )  ) dy \leq  \theta, \nonumber\\
	\nonumber P(y) \text{ is decreasing and right continuous}.  
\end{gather}
Before presenting the result, let me introduce the following  multiplier:
\begin{gather*}
	\mu^* = \max_{p\in (0,1]}  \frac{p}{w_-(p)}  
\end{gather*}

As illustrated in Figure \ref{pic2}, $\frac{1}{\mu^*}$ is the slope of the tangent line. The tangent point is clearly unique and interior. I summarize these observations in the following lemma.
\begin{lemma}
	\label{lemm1}
	The multiplier $\mu^*$ exceeds one, and there exists a unique $p^*\in(0,1)$ such that $\mu^*  w_-(p^*) = p^*$.
\end{lemma}
Now it is straightforward to establish an upper bound for the expected revenue:
\begin{align*}
	\int_0^{\infty} P(y) dy  &=  \int_0^{\infty} w_-(P(y))  \frac{P(y)}{w_-(P(y))} dy\\
	&\leq  \int_0^{\infty} w_-(P(y))  \mu^* dy  \leq \frac{\mu^*\theta}{\lambda }.
\end{align*}
In addition, this upper bound can be achieved by the following binary random variable.
\begin{proposition}
	\label{prop1}
	The maximum profit in program \ref{program1} is $\mu^* \theta / \lambda $ and the unique optimal tail distribution of price $P^*(y)$  is a binary random distribution given by:
	\begin{gather*}
		P^*(y)  = \begin{cases}
			p^* \quad   &\text{ if }   y < \frac{\mu^* \theta}{ p^* \lambda},\\
			0 &\text{ if }   y\geq \frac{\mu^* \theta}{ p^* \lambda}.
		\end{cases}   
	\end{gather*}
\end{proposition}

The optimal binary random price equals $0$ with probability $1-p^*$ and equals $\frac{\mu^* \theta}{ p^* \lambda} > \frac{\theta}{\lambda}$ with probability $p^*$. Because $p^*$  minimizes the ratio $w_-(p)/p$, the large payment happens with the probability that the buyer under-evaluates the most.
In comparison, if the seller uses a deterministic price, his revenue is simply $\frac{\theta}{\lambda}$. The optimal binary random price simply multiplies the expected revenue by $\mu^*$.

To achieve the upper bound characterized in Proposition \ref{prop1} one can set a binary random stopping time $\tau_0$ such that $\mathbb{P} (\tau_0=1) = 1-p^*$ and $\mathbb{P} (\tau_0=3) = p^*$. Then setting a cumulative price path such that $p(t)= \frac{\mu^* \theta}{ p^* \lambda}1_{t>2} $.
When $\tau_0=1$, the buyer gets the product for free, while when $\tau_0=3$, the buyer pays a large payment $\frac{\mu^* \theta}{ p^* \lambda}$ that exceeds her willingness to pay. This pricing process does not satisfy DIR conditions, as the buyer will simply turn away after $t>1$.

\subsection{Make DIR Tractable: Exponential Decomposition}

The main difficulty of the optimal design problem comes from dynamic inconsistency, which forces us to track of all the posteriors simultaneously. This is particularly intractable since it is unclear how to reduce the posterior belief to some low-dimensional sufficient statistic, such as posterior means, in the presence of the rank-dependent probability weighting function. Furthermore, normalization of the belief is no longer feasible in my model. To see this more clearly, consider the following:
\begin{lemma}
	\label{Lemma_eventual}
	Any DIR-compatible $(\tau,T(t))$ such that
		$P(\tau < \infty) < 1$ is not optimal.
\end{lemma}

Lemma \ref{Lemma_eventual} states it is without loss of generality to focus on mechanisms where the buyer eventually obtains the good for certain.
Then, DIR can be expressed  as
\begin{gather*}
	\theta - \lambda \int_0^{\infty} 	w_- \big(  \frac{P( T (\tau) - T(s) > y   )  }{ 1- P(\tau \leq s)}  \big) dy \geq 0, \quad \forall s,~ P(\tau > s) > 0.
\end{gather*}

In the standard EU framework where $w_-(p)=p$, one can multiply both sides by $(1- P(\tau \leq s) ) $ so that the expression becomes linear in $P(\tau <t)$. The nonlinearity of $w_-$ makes it difficult to track the impact of changing posteriors.

To restore tractability I proceed in two steps.
First, I show that  I can transform the joint design problem of  $(\tau,T(t))$ into a design problem of a random total payment $T$ in which the buyer can potentially choose to quit at any payment $t< T$. 
\begin{lemma}
	\label{lemma4}
	For any dynamic pricing process $(\tau,T(t))$, if $P(\tau<\infty)=1$ and  $T(\tau)$ solves the following auxiliary problem
	\begin{gather*}
		T(\tau) \in  \argmax_{T \in \mathcal{T}}  \int_0^{\infty} P(T>y) dy, \\
		s.t. \quad \lambda \int_0^{\infty} w_-(  P(T-s > y) | T> s)  dy \leq \theta, \quad \forall s, ~  s.t. ~ P(T > s) > 0,
	\end{gather*}
then $(\tau,T(t))$ is an optimal dynamic pricing process.
\end{lemma}

$\mathcal{T}$ is the set of all nonnegative random variables whose CDF admits a regular decomposition. Lemma \ref{lemma4} indicates that the two design objectives in the original problems are somewhat redundant: what matters is the composition $T(\tau)$. However, transforming the problem into this ``basic form'' per se does not simplify the problem at all, as the problem is still nonlinear in $P(T>y)$ and the intractability problem remains.

In fact, the \emph{key technical insight} lies in the second step: the problem becomes considerably more tractable once  I decompose $T$ in the auxiliary problem back into $(\tau_0,T(t))$, where $\tau_0$ is a fixed exponential stopping time,
\begin{gather*}
	P(\tau_0 > t) = e^{-t}.
\end{gather*}

\begin{proposition}
	\label{prop3}
	For any feasible $T$ in the auxiliary   problem, there exists a (regular) cumulative payment function $T(t)$ such that $(\tau_0,T(t))$ is DIR compatible in the original design problem and 
	\begin{gather*}
		P( T \leq t )  =  P( T(\tau_0) \leq t), \quad \forall t \in \mathbb{R}^+.
	\end{gather*}
\end{proposition}

Now, I show how the exponential stopping time $\tau_0$ can help me simplify the DIR constraint:
	\begin{gather*}
  \lambda \int_{0}^{\infty} w_-( P( T (\tau_0) - T(s) > y  | \tau_0 >s )  ) dy \leq \theta, \quad   \forall s \geq 0.
	\end{gather*}
 
	Define 
	\begin{gather*}
		T^{-1}(x) = \sup \{ t \in \mathbb{R}^+ |T(t) \leq x   \}.
	\end{gather*}
 
Since $T$ is left continuous, $T(t) \leq x \Leftrightarrow t \leq T^{-1}(x)$. Consequently,
	\begin{align*}
		&\lambda \int_{0}^{\infty} w_-( P( T (\tau_0) - T(s) > y  | \tau_0 >s )  ) dy \\
		&= \lambda \int_{0}^{\infty} w_-( P( \tau_0   > T^{-1}( y + T(s))  | \tau_0 >s )  ) dy \\
		&= \lambda \int_{0}^{\infty} w_-( e^{-T^{-1}( y + T(s))+s }    ) dy .
	\end{align*}
 
	For intuition, suppose that $T$ has  positive derivatives everywhere and so does $T^{-1}$. Then, I can further transform the equation as
	\begin{gather*}
		 \lambda \int_0^{\infty} w_-(  e^{-T^{-1}( y + T(s))+s }    )  \frac{1}{T^{-1'}( y + T(s))  }dT^{-1}( y + T(s)),\\
		= \lambda \int_{s}^{\infty} w_-( e^{-t +s }) T'(t) dt.
	\end{gather*}
 
	In general, $T(t)$ might have some jump points, and I prove in the Online Appendix that:
	\begin{lemma}
		\label{lemma5}
		\begin{gather*}
			\lambda \int_{0}^{\infty} w_-( e^{-T^{-1}( y + T(s))+s }    ) dy  
			= \lambda \sum_{i,t_i \geq s} T_i w_-( e^{- t_i + s} ) + \lambda 	\int_{s}^{\infty} w_-( e^{-t +s }) f(t) dt , \quad \forall s\geq 0.
		\end{gather*}
	\end{lemma}

	With the help of Lemma \ref{lemma5}, I can simplify the exponential design problem as
	\begin{gather*}
		V= \max_{f,T_i,t_i\geq 0}  ~  \sum_i T_i e^{-t_i} + \int_0^{\infty}e^{-t}  f(t) dy  ,\\
		s.t.  \quad   \lambda \sum_{i,t_i \geq s} T_i w_-( e^{- t_i + s} ) + \lambda 	\int_{s}^{\infty} w_-( e^{-t +s }) f(t) dt \leq \theta , \quad \forall s\geq0.
	\end{gather*}

	In this representation of the problem, the objective becomes linear in the design objectives $(f,T_i)$, and I solve it with a weak-duality approach in the next subsection.

The technique that I just presented is not restricted by the specific structure $(\tau,T(t))$ of the model. Instead, it  has general applications in design problems under rank-dependent preferences.  Even if one starts with problems such as the auxiliary  problem in Lemma \ref{lemma4},\footnote{In Section 6.3, I present two other economic applications of this auxiliary problem.} where the design objective is just a random variable $T$, it could still be beneficial to introduce the decomposition purely as an artificial mathematical construct. Mathematically, decomposing $T$ into $(\tau_0,T(t))$ is a specific way of changing coordinates (and consequently changing the representation) of $T$ that works well under rank-dependent preferences. This is because when  changing $P(T>y)$ or the density $\rho(T=y)$ at $y$, both the value realization and the rank of the value realization in the whole distribution change jointly, which makes it difficult to track its impact. In contrast, when changing $T(t)$ at $t$ with fixed $\tau_0$, only the value realization changes, and the rank of the whole distribution remains unchanged.

\subsection{Optimal Exponential Design}
In this subsection, I sketch the proof for the uniquely optimal cumulative payment function $T^*(t)$ for the exponential stopping time $\tau_0$. The economic intuition for the optimal mechanism is presented in the next subsection.
\begin{proposition}
	\label{prop_expo}
	There exists a uniquely optimal $T^*(t)$ for $\tau_0$ such that
	\begin{gather*}
		T^*(t) =\frac{\theta}{\lambda}  \Big(     1_{t>t_0}   + \int_0^{\min \{ t,t_0 \}}  f(s) ds  \Big),
	\end{gather*}
where $f$ is the unique function in $L_1[0,t_0]$ that ensures that DIR is binding in $[0,t_0]$:
\begin{gather*}
	w_-(e^{-t_0 + t})  +  \int_t^{t_0} w_- (e^{-s+t}  ) f(s) ds  = 1, \quad \forall t \in [0,t_0].
\end{gather*}
The constant $t_0$ is determined by
\begin{gather*}
		t_0 = \min \{  t \in [0, - \log p_1] \mid   g_0(t) \leq 0 \}, 
\end{gather*}
where $g_0(t)$ is the unique function in  $L_1[0,-\log p_1]$ such that
	\begin{gather*}
	e^{-t} - w_-(e^{-t}) - \int_0^t  w_-(e^{-t+s}) g_0(s) ds = 0, \quad \forall t\in [0, - \log p_1].
\end{gather*} 
\end{proposition}

The optimal cumulative payment $T^*(t)$ remains constant when $t>t_0$, which is equivalent to saying that the buyer obtains the good for certain after time $t_0$. $T^*(t)$ has a jump of $\theta/\lambda$ at time $t_0$, which means that if the buyer still does not receive the good by time $t_0$, she is charged a lump sum payment of $\theta/\lambda$ to obtain the good for certain. At this final time $t_0$, the buyer is indifferent between quitting and obtaining the good at a deterministic full price. At any time $t < t_0$, the buyer will be charged a flow payment that keeps her indifferent between staying and quitting.

\noindent
\paragraph{Sketch of Proof}
	The formal proof is in the Appendix. 
	The exponential design problem is:
	\begin{gather*}
		V= \max_{f,T_i,t_i}   \sum_i T_i e^{-t_i} + \int_0^{\infty}e^{-t}  f(t) dy  ,\\
		s.t.  \quad   \lambda \sum_{i,t_i \geq s} T_i w_-( e^{- t_i + s} ) + \lambda 	\int_{s}^{\infty} w_-( e^{-t +s }) f(t) dt - \theta \leq 0, \quad \forall s\geq0.
	\end{gather*}

	To solve it, I ``construct'' a Lagrange measure $\gamma(t)$ for the DIR constraints and transform the problem into an unconstrained problem. The  Lagrange measure is supported in $[0,t_0]$, and its  ``construction'' ensures that in the relaxed problem, the linear coefficient of $f$ and $T_i$ is zero pointwise on $[0,t_0]$ and negative on $(t_0,\infty)$. Then, on the one hand, optimality requires $f$, $T_i$ to be zero in $(t_0,\infty)$ and imposes no requirement on their value in $[0,t_0]$. On the other hand, complementary slackness requires DIR  to be binding in $[0,t_0]$. These two jointly determine the unique shape of $T^*(t)$.
	I place quotation marks on ``construction'' because I do not explicitly define the Lagrange measure: I only prove it exists as the solution of a Volterra integral equation (the integral equation for $g_0$ in the proposition).

\subsection{Optimal Mechanisms}
The following theorem concludes the previous subsections.

\begin{theorem}
	\label{thm3}
	$(\tau_0, T^*(t))$ described in Proposition \ref{prop_expo} is an optimal dynamic pricing process. In addition, a dynamic pricing process $(\tau,T)$ is optimal if and only if
	\begin{gather*}
		P(\tau < \infty) = 1, \\
		P(T(\tau) \leq y) =P(T^*(\tau_0) \leq y), \quad \forall y \in \mathbb{R}^+
	\end{gather*}
\end{theorem}

\begin{figure}[h]
	\centering
	\begin{subfigure}{.47\textwidth}
	\centering
\includegraphics[width=8cm]{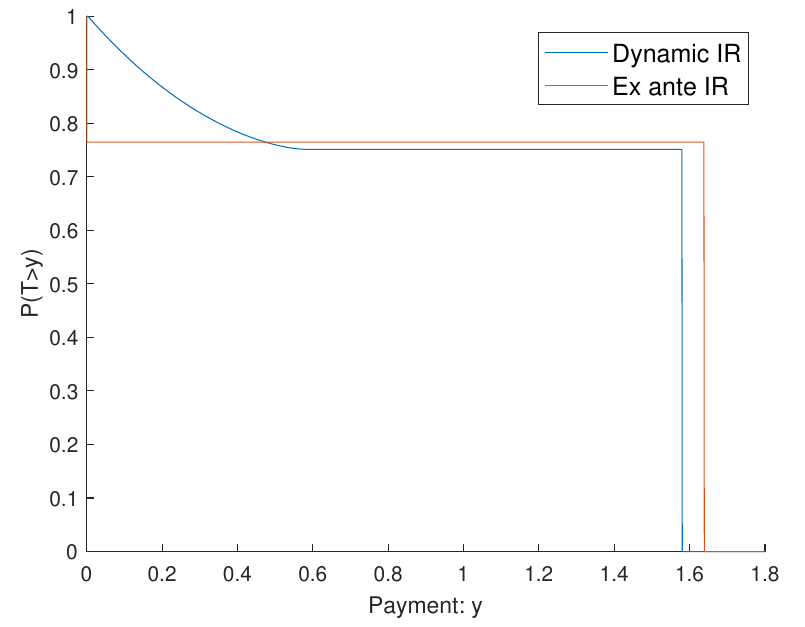}
\caption{Tail Distributions of Total Payment} 	\label{pic7}
	\end{subfigure}
	\begin{subfigure}{.47\textwidth}
	\centering
\includegraphics[width=8cm]{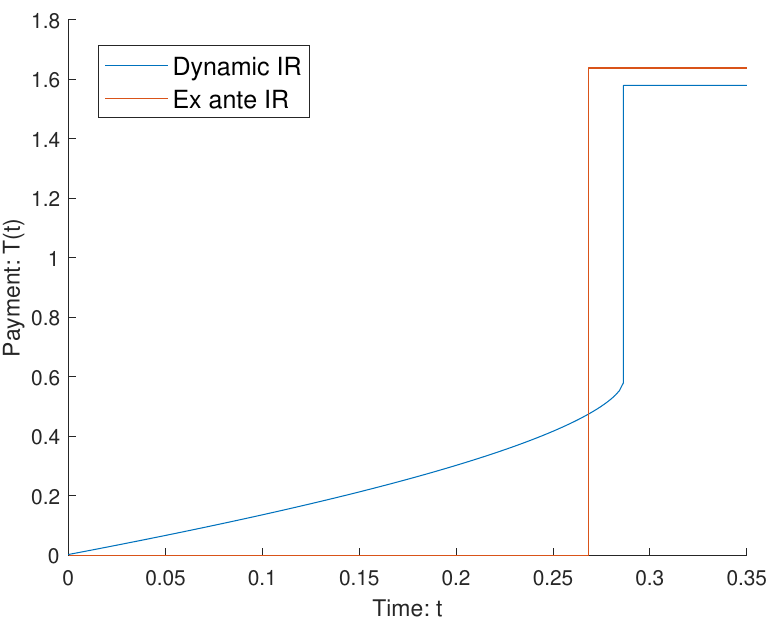}
		\caption{Payment Function $T(t)$ under $\tau_0$} 	\label{pic11}
	\end{subfigure}
	\caption{Optimal Prices with Ex Ante/Dynamic IR Constraints ($\gamma=0.65$).}
	\label{fig:CommitNonCommit}
\end{figure}

According to Theorem \ref{thm3}, the distribution of total payment $T(\tau)$ for any optimal dynamic pricing process $(\tau,T(t))$ is unique, so I can plot it in Figure \ref{pic7} and compare it with the unique optimal distribution with ex ante IR. The red curve  represents the tail distribution of the binary distribution characterized in Proposition \ref{prop1}, and the blue curve represents the tail distribution of $T^*(\tau_0)$ in Theorem \ref{thm3}. The blue curve has a flat region, which  represents the final lump sum payment that the buyer has to pay if her purchase process does not end earlier.

To deliver some intuition about the structure of the mechanism, it is helpful to focus again on the exponential stopping time $\tau_0$.
Figure \ref{pic11} plots the optimal cumulative payment function for the two different IR constraints, under the exponential stopping time $\tau_0$. The blue curves represent the optimal payment function with DIR constraints, characterized in Proposition \ref{prop_expo}. The red curve  represents the optimal payment function with the ex ante constraint, characterized in Proposition \ref{prop1}. To replicate the binary distribution in Proposition \ref{prop1}, the red curve simply charges a lump sum payment $T^* = \frac{\mu^* \theta}{\lambda p^*}$ at time $t^* = -\log p^*$. From an ex ante perspective, this lump sum payment happens with probability $P(\tau_0>t^*) = p^*$.

From the previous subsection, I know that with exponential stopping time, the expected revenue can be expressed as
\begin{gather*}
	 \sum_i T_i e^{-t_i} + \int_0^{\infty}e^{-t}  f(t) dy, 
\end{gather*}
and the ex ante IR can be expressed as
\begin{gather*}
 \lambda \sum_{i,t_i \geq 0} T_i w_-( e^{- t_i } ) + \lambda 	\int_{s}^{\infty} w_-( e^{-t }) f(t) dt \leq \theta.
\end{gather*}
From an ex ante perspective, charging one unit of money at time $t$ increases the expected revenue by $e^{-t}$ but incurs a loss of $\lambda w_-(e^{-t})$ for the buyer. Intuitively, it is more efficient for the seller to charge the buyer at time $t$ where the ratio $\frac{w_-(e^{-t})}{e^{-t}}$ is small. This ratio is minimized at $t=t^*$. Thus, if only the ex ante IR condition is concerned, the seller allocates all the payment at time $t^*$, as characterized in the red curve.

How would the seller adjust the payment function in the presence of DIR constraints? For intuition, suppose that seller can only charge the buyer over time $[0,t^*]$. Note that the ratio $\frac{w_-(e^{-t})}{e^{-t}}$ is monotonically decreasing over $[0,t^*]$ because of the curvature of $w_-$. Therefore,  for any $t_1 < t_2 \leq t^*$, the seller always benefits from moving the payment from $t_1$ to $t_2$ as long as the participation constraints are not violated. Intuitively, the seller should charge as much as possible at $t^*$, so he charges a lump sum payment equal to the buyer's value. Then, starting from $t^*$ to $0$, he should allocate as much payment to $t$ that is closer to $t^*$ as he can. This intuitively explains why the DIR constraints are always binding at every time. 

In the optimal solution, the final lump sum payment is not charged exactly at $t^*$ but rather at some $t$ slightly larger than $t^*$. This is because the optimal mechanism necessarily contains both the final lump-sum payment and the earlier flow payment. Shifting the whole tail distribution curve to the right of $t^*$ can allow for a better ratio for the flow payment, at the cost that the lump sum payment enjoys a worse ratio. The benefit of such a right shift is of first order, but the cost is of second order at $t^*$.

\subsection{Pricing Processes in Practice}
\paragraph{Gacha Games Industry}

The industry of gacha games displays a distinct trajectory in its pricing practices. Initially, these games employed a stationary pricing process: both the cost of each loot box and the probability of obtaining the desired product from it remained constant in the entire process. However, in 2016, a notable shift occurred when Granblue Fantasy introduced the "hard pity system." Under this system, the pricing for the loot box remains constant for the initial N purchases (300, in this game's case). If the buyer hasn't obtained the desired product after these 
N purchases, they are guaranteed the product at no additional cost. Most gacha games that followed incorporated this system into their pricing strategy.

Figure \ref{pic9} plots the tail distribution of the total payment in three different  mechanisms that are optimal within their classes. The yellow curve represents the stationary mechanism (without the pity system), the blue curve represents the hard pity mechanism, and the red curve represents the optimal mechanism. The area below the curve represents the expected profit for the seller.

If one assumes that the convergence of mechanisms indicates that consumers  prefer the latter mechanism, and firms use it to extract more surplus. 
Then, the evolution of the pricing processes, in which the total payment moves from a geometric distribution to a truncated geometric distribution, indicates that consumers are not simply risk-loving. They prefer a certain form of randomness but do not want to be exposed to a negatively skewed payment.This observed behavior aligns with the expectations set by probability weighting.

In addition to the qualitative change, there is also an evolution in tuning parameters. For example, in newer gacha games, the probability of reaching the maximum number of purchases gradually increases.
\begin{figure}[htpb]
	\centering
	\begin{subfigure}{.47\textwidth}
		\centering
	\includegraphics[width=8cm]{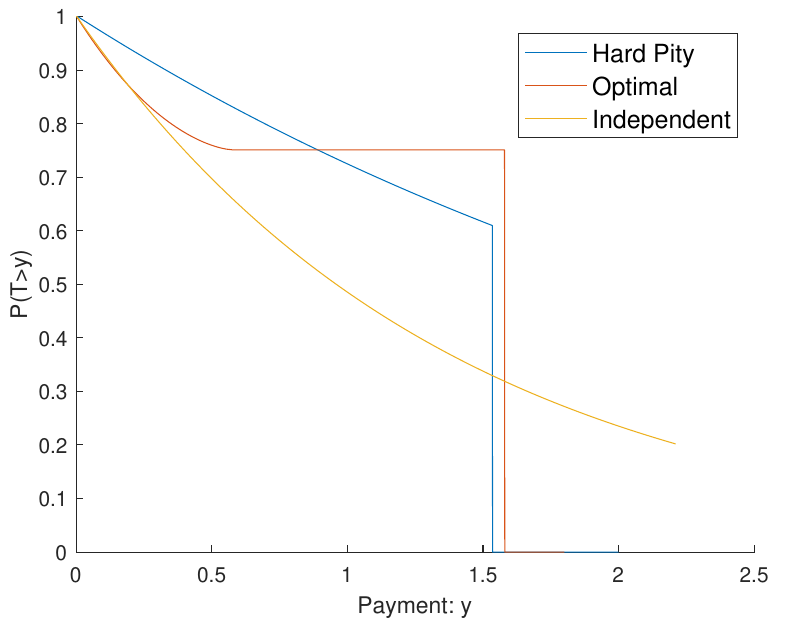}
	\caption{Optimal Mechanisms within Classes} 	\label{pic9}
		\hspace{0.1cm}
	\end{subfigure}
	\begin{subfigure}{.45\textwidth}
		\centering
\includegraphics[width=8cm]{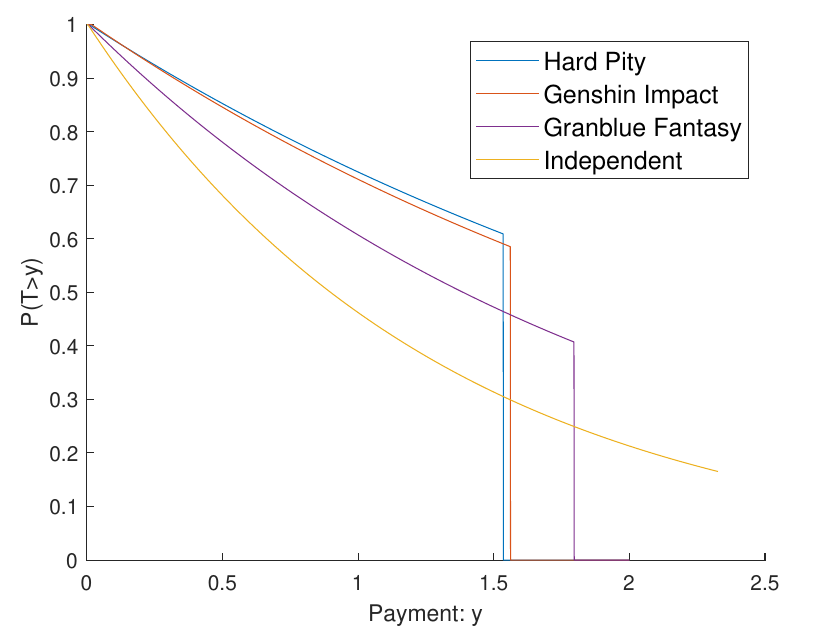}
\caption{Evolution in Practice} 	\label{pic10}
\hspace{0.1cm}
	\end{subfigure}

	\caption{Tail Distributions of Various Mechanisms $(\gamma=0.65)$.}
	\label{fig:optimal}
\end{figure}
 Figure \ref{pic10} plots the tail distribution of the total payment in Granblue Fantasy (2016) and Genshin Impact (2020) and compares them to the optimal hard pity mechanism and the optimal stationary mechanism (the mechanism adopted before 2016). As the figure illustrates, over time, the mechanisms adopted in practice seem to converge to the optimally tuned hard-pity mechanism.\footnote{Due to the norm of the industry, it is almost impossible for one game to change its monetization mechanism in its lifetime, so the modification has to be done in newer games.}

I also numerically calculate the maximum profit within the three classes of mechanisms with different parameter values $\gamma$ (of the parameterized probability weighting function). The results are summarized in Table \ref{table1},
where I normalize $\theta/\lambda =1$ so that with deterministic pricing the seller obtains a profit of 1. 
As Table \ref{table1} suggests, the skewness of the geometric distribution harms the buyer to such an extent that the performance of independent chests is even worse than the deterministic price. On the contrary, the truncated geometric distribution induced by the hard pity system captures 90\% of the increase in profit compared to the optimal mechanism. This might seem surprising to some readers because ``charge full price at the end'' and ``charge nothing at the end'' appear to be totally different. To understand this, note that when forgoing the full price charge at the end, it is optimal for the seller to both increase the payment for each earlier loot box and reduce the probability that each box will deliver. After optimally tuning the parameters, both mechanisms assign   intermediate probability numbers, which are concentrated around $p^*$, to  the majority of their payment realizations. It is true that the tail probability in the optimal mechanism is more concentrated at $p^*$ than the hard pity mechanism, but the effect of such reduced dispersion around $p^*$ is of second order.

\bigskip

\begin{table}[h!]
	\centering
	\begin{tabular}{|c|c|c|c|c|c|c|}
		\hline
		$\gamma$	& 0.50 & 0.55 & 0.6 & 0.65 & 0.70 & 0.75  \\
		\hline
		Independent	& 0.8024 & 0.8152 & 0.8302 &  0.8417  & 0.8656 &  0.8855\\
		Hard Pity	& 1.5211 & 1.3841 & 1.2847 &  1.2106  & 1.1541 &  1.1106\\
		Optimal	& 1.5758 & 1.4265 &1.3178  & 1.2364 &  1.1740   & 1.1256\\
		\hline
	\end{tabular}
	\caption{Profit of Optimal Mechanisms within Each Class}
	\label{table1}
\end{table}

There are many different forms of simple discretized mechanisms that can better approximate the optimal mechanism. For example, Figure \ref{pic8} represents the tail distribution of a simple modification of the hard pity system where at the end, the product is sold at full price instead of given for free. However, from my communication with practitioners in the industry, there are other considerations such as marketing and public relationships when the pricing process is designed. As a leading example, the relationship between players and game companies is highly  sensitive and intense with respect to monetization, as the player community tends to interpret any modification to the pricing mechanism as an attempt at further exploitation. This might explain why the transition from stationary loot boxes to the hard pity system was welcomed in reality since the concept of a free product can be naturally interpreted as generosity. This also implies that if one attempts to asses the implications of the optimal mechanism in reality, the choice of an appropriate approximation format must be delegated to experts in the industry.

\begin{figure}[h!]
	\centering
	\includegraphics[width=8cm]{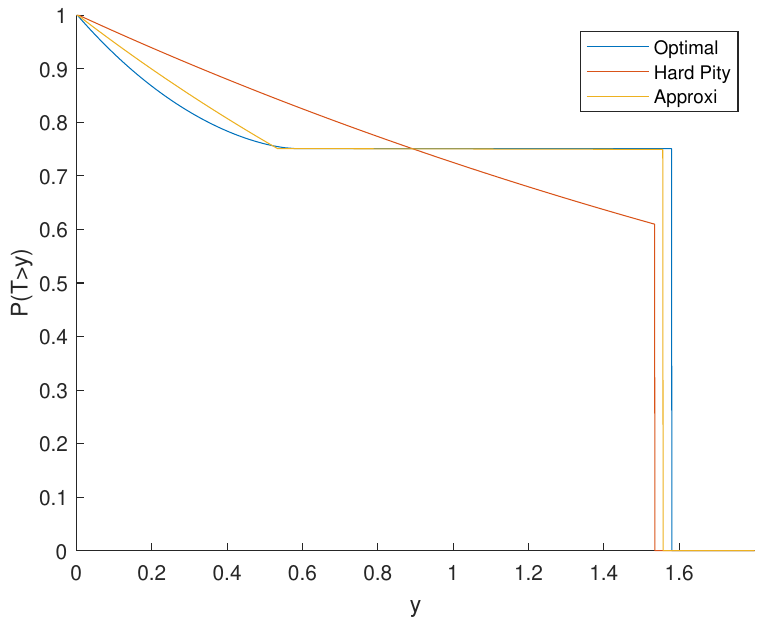}
	\caption{Tail Distributions of Three Mechanisms $(\gamma=0.65)$.} 	\label{pic8}
\end{figure}

\paragraph{Random Prices in Other Industries} 
Similar pricing processes appear in other entertainment industries such as  card collection,  crane machines, and blind boxes. Penny auction,  as a new online-auction format, also presents similar feature to loot boxes.
Why, then, aren't random prices more widespread? Firstly, many consumers maintain a firm moral opposition to random payments, viewing them as akin to gambling and an exploitation of human vulnerabilities. In contrast, consumers in entertainment sectors display a greater tolerance for stochastic pricing. This perhaps elucidates why such pricing predominates in the entertainment domain. Secondly, implementing a random pricing scheme demands substantial commitment from sellers. In the realm of gacha games, legal regulations in many countries mandate that companies transparently disclose their probability designs. Furthermore, authorities have the power to detect and penalize any discrepancies.\footnote{For example, in 2018, three South Korean developers faced significant fines from the country’s Fair Trade Commission for misleading players about the odds in their in-game loot boxes. In 2022, Paradise M Taiwan’s representative Game Juzi incurred a penalty of 2 million yuan for manipulating probabilities.} Additionally, third-party applications can aggregate data from vast player bases to cross-check the declared probability designs.\footnote{For instance, \url{https://github.com/xunkong/KeqingNiuza} offers an application for Genshin Impact that serves this purpose.}


\section{Modeling Choices and Robustness}

\subsection{Robustness: Rank-Dependent Expected Utility}
Prospect theory is not the only decision-theoretic model that embeds probability weighting. A more ``conventional'' framework is the RDEU of \cite{QUIGGIN1982}, in which the buyer's payoff is
\begin{align*}
	\text{RDEU}(X,T,\theta) &= \int (\theta X -T) ~ d(w\circ F) \\
	&= \int_0^{\infty} w(P( \theta X - T > y  )) dy + \int_{-\infty}^0 (w(P( \theta X -T  >y  )) -1 )dy,
\end{align*}
where $F$ is the cumulative distribution function of the ex post payoff $\theta X -T$. When $w(p)=p$, RDEU reduces to the standard expected utility. Compared to CPT, the RDEU framework is more ``conventional'' because there is no framing effect. The agent only cares about the absolute level of utility instead of calculating  gains and losses relative to a reference point.

It turns out that the only driving force in my framework is the probability weighting. Proposition \ref{prop1} and Theorem \ref{thm3}  also characterize the optimal mechanisms under RDEU with an appropriate normalization.
\begin{proposition}
	\label{robust_RankDependent}
	Define $w_-(p) =1-w(1-p)$, and let $\lambda=1$; then,  Proposition \ref{prop1} and Theorem \ref{thm3} characterize the optimal mechanisms in the RDEU framework. 
\end{proposition}

\subsection{The Modeling Choice within Cumulative Prospect Theory}
In CPT, outcomes are evaluated by a utility function $U$  relative to a reference
point that separates all outcomes into gains and losses.
In addition, there are weighting functions  distorting the cumulative probability. The CPT utility of general random variables $Y$ is defined as:
\begin{gather*}
	CPT(Y) = \int_{\mathbb{R}^+} w_+ (P(U(Y)>y) dy - \int_{R^-} w_- ( P(  U(Y)<y ) ) dy.
\end{gather*}

For example, suppose that $Y$ is the money that a gambler gains from a lottery; \cite{TVERSKYKAHNEMAN92} assumed the following functional form for $U$:
\begin{gather*}
	U(y) = \begin{cases}
		y^{\alpha}  \text{ if } y \geq 0,\\
		-\lambda (-y)^{\beta}   \text{ if } y<0. 
	\end{cases}    
\end{gather*}
For tractability and to provide a closed link to the standard literature in pricing and mechanism design, this paper assumes a linear monetary utility function by letting $\alpha=\beta=1$.

In prospect theory literature, there is more than one way to define the reference point. In this paper, I choose the ex ante condition as the reference point. In the dynamic model, since the utility of money is linear and I assume separate accounting (to be discussed below), whether the reference point is the initial condition at time $0$ or the initial condition at the time of the current decision is irrelevant.

Another modeling choice to make is the following: does the buyer
experience gain-loss utility separately
from money and from the product being sold
(separate evaluation) or jointly from the net
utility of the transaction (net evaluation). As \cite{2014Botond} argues,
``in commodity auctions---where the payment is monetary and the product is nonmonetary---the former assumption seems more appropriate, while in induced-value laboratory
experiments---where both the payment
and the “product” are monetary---the latter
assumption seems to apply.''
This paper chooses the first assumption that the buyer uses two separate accounts (the robustness of the results under joint accounting appears in the next subsection).  Under this assumption, the product account only has gains while the payment account only has losses.

In conclusion, denote $\theta$ as the good's value for the buyer. The payoff of a menu $(X,T)$ for a buyer of type $\theta$ can be expressed as
\begin{align*}
	CPT(\theta,X,T) &= \int_{0}^{\infty} w_+(P(\theta X > y)) dy - \int_{-\infty}^0 w_-( P( -\lambda T < y )  ) dy,\\
	&=\theta w_+ (P(X  =1 )  ) - \lambda \int_{0}^{\infty} w_-( P(  T > y )  ) dy.
\end{align*}

\subsection{Robustness: Joint Accounting}
One alternative modeling choice is to assume that the CPT buyer evaluates gains and losses in one joint account. Under this specification, the payoff is
\begin{gather*}
	\overline{CPT}(\theta,X,T) = \int_{0}^{\infty} w_+(P(\theta X -T > y)) dy - \int_{-\infty}^0 w_-( P( \lambda (\theta X -T) < y )  ) dy.
\end{gather*}
The most significant difference between the two specifications is that the correlation between $X$ and $T$ matters under joint accounting but not under separate accounting. However, this difference is irrelevant to the main results of this paper because the seller always sells the good to the buyer eventually, namely $X=1$, in which case there is no further role for correlation.\footnote{I thank Philipp Strack for this insight. The main results refer to Proposition \ref{prop1} and Theorem \ref{thm3}. The detailed characterization of Theorem \ref{thm2} is not robust, but the optimal pricing process will still be stationary.}
Formally, I can establish the following robustness result:
\begin{proposition}
	\label{prop5}
	Suppose that $\lambda=1$, $w_-(p)= w_+(p)=w(p)$ and $w(p) + w(1-p)=1$. Then, Proposition \ref{prop1} and Theorem \ref{thm3} also hold under joint accounting.
\end{proposition}
Proposition \ref{prop5} states that the optimal mechanisms under joint accounting \emph{are exactly the same as} the optimal mechanisms under separate accounting (characterized in Proposition \ref{prop1} and Theorem \ref{thm3}), when there is no loss aversion, the probability weighting of gains and losses is the same, and probability weighting functions are symmetric.

Of course, from an empirical perspective, the assumptions in Proposition \ref{prop5} are not very realistic. However, I regard this result as a very strong signal suggesting that in the framework of this paper, the difference between the two specifications is of second order.

\section{Naive Buyer}
As the previous sections indicate, the stationary mechanism performs poorly when the buyer is sophisticated. Then, is there any rationale for using the stationary mechanism in the early stage of gacha games?
In this section, I consider a discrete-time model and show that the optimal dynamic mechanism takes a stationary form when the buyer is naive. 

\paragraph{Discrete-Time Model}
There are discrete periods $t=1,2,...,\infty$. The seller designs the probability of delivering the good and how much to charge in each period. Formally, he designs $(X_t,T_t)$, where the random variable $X_t \in \{0,1\}$ indicates whether the buyer obtains the good in period $t$, and $T_t$ is the random price that the buyer has to pay in period $t$. The distribution of $(X_t,T_t)$ can be history dependent.

\paragraph{Timing}
In each period $t$, the buyer first decides whether to purchase, $A_t \in \{ 0,1\}$. Then, both $T_t$ and $X_t$ realize such that  the buyer pays $A_t T_t$  and receives the good if $A_t X_t =1$. I use $\tau = \min\{t| A_t X_t =1 \}$ to denote the time that the buyer receives the good. The game ends at $\tau$, and for notational simplicity I denote  $T_{t} = 0$ for any $t > \tau$.


\paragraph{History and Strategies}
$H_t = \{ h_t | h_t = (T_0,X_0,A_0...,T_{t-1},X_{t-1},A_{t-1}  )  \}$ is the set of all histories up to period $t$. $\sigma^t_s(h_t,h_s): H_t \times H_s \to \{0,1\}$ is the action that the buyer believes that she will take  if she faces $h_s$ in future period $s\geq t$, when she will actually face $h_t$. A strategy $\sigma^t (h_t) = (\sigma^t_s(h_t, \cdot ) )_{s\geq t}  \in S^t$ is the strategy that the buyer believes she will play thereafter at the start of period $t$.  For notational simplicity, I use $\sigma^t_s(h_s)$ to denote $\sigma^t_s(h_t,h_s)$ when there is no confusion. Naturally, the action the buyer actually takes in period $t$ is $A_t = \sigma^t_t(h_t)$.

\paragraph{Dynamically Inconsistent Preferences}
At the start of each period $t$, the buyer will evaluate a strategy $\sigma^t$ according to her updated belief of $(X_s,T_s)_{s\geq t}$:
\begin{gather*}
\int_0^{\infty} w_+ (  P( \delta^{\tau-t} \theta > y| h_t, \sigma^t ) ) dy - \lambda \int_{0}^{\infty} w_- ( P (  \sum_{s=t}^{\tau} \delta^{s-t} \sigma^t_s(h_s) T_s > y | h_t, \sigma^t ) ) dy.
\end{gather*}
The specification, which builds on separate evaluation and quasilinear utility, is the same as in the previous section. $\delta^{\tau-t} \theta$ represents the discounted gains from the good, while $ \sum_{s=t}^{\tau} \delta^{s-t} \sigma^t_s(h_s) T_s$ is the discounted loss from payment.

\paragraph{Design Problem}
The design problem can be written as:
\begin{gather}
	\max_{(X_t,T_t),\sigma} \quad \mathbb{E}  \sum_{t=1}^{\infty}  \delta^t     T_t \sigma_t^t(h_t), \\
		s.t. \nonumber \\
		\text{ (SIR) } \int_0^{\infty} w_+ (  P( \delta^{\tau-t} \theta > y| h_t, \sigma^t ) ) dy - \lambda \int_{0}^{\infty} w_- ( P (  \sum_{i=t}^{\tau} \delta^{s-t} \sigma^t_s(h_s) T_s > y | h_t, \sigma^t ) ) dy \geq 0 ~ \forall h_t, \nonumber\\
			\int_0^{\infty} w_+ (  P( \delta^{\tau-t} \theta > y| h_t, \sigma^t ) ) dy - \lambda \int_{0}^{\infty} w_- ( P (  \sum_{i=t}^{\tau} \delta^{s-t} \sigma^t_s(h_s) T_s > y | h_t, \sigma^t ) ) dy  \nonumber \\
			\geq 		\int_0^{\infty} w_+ (  P( \delta^{\tau-t} \theta > y| h_t, \bar{\sigma}^t ) ) dy - \lambda \int_{0}^{\infty} w_- ( P (  \sum_{i=t}^{\tau} \delta^{s-t} \bar{\sigma}^t_s(h_s) T_s > y | h_t, \bar{\sigma}^t ) ) dy,  ~\forall h_t,~ \bar{\sigma}^{t}.	\nonumber
\end{gather}
Here $\sigma = (\sigma^t)_t$ denotes a sequence of the buyer's strategy $\sigma^t$. At each time period $t$, the buyer believes that she will play strategy $\sigma^t$ from period t onward. The incentive constraints require that, given her updated belief in period $t$,  the  strategy $\sigma^t$ should be optimal if it can indeed be implemented by her future self. In particular, $\sigma^t$ is better than immediately quitting the market forever. It turns out that this constraint is sufficient to derive the optimal mechanism, so I separate it out with a special notation SIR (sequential individual rationality).

I use the term SIR to distinguish it from the DIR used in Section 3. In the setting of this section, buyer can commit to make a random payment within each period, and the incentive constraints in each period only need to hold before the random outcome realizes. Thus, if the buyer is sophisticated, the static mechanism in the benchmark solution (Proposition \ref{prop1}) is optimal. The dynamic design is purely intended to exploit the naivety of the buyer.

To present the result, denote the set of on-path histories as $H^* = \{h_t | \tau\geq t,  A_s = 1 ~ \forall  s <t\}$. $H^*$ is the set of histories where the buyer continued purchasing in the past but has not received the good.
\begin{theorem}
	\label{thm2}
	The optimal profit is
	\begin{gather*}
	V^* = \max_{x\in [0,1]} \frac{\mu^* \theta}{\lambda} \frac{w_+(x)}{1-(1-x)\delta}.
	\end{gather*}
	The following dynamic pricing process $(X_t^*,T_t^*)$ is optimal. 
	\begin{gather*}
		\begin{cases}
				P(X_t^*=1 |h_t \in H^*) &= x^*\in \argmax  \frac{w_+(x)}{1-(1-x)\delta}, \\
				P(X_t^*=1 |h_t \not\in H^*) &= 0.
		\end{cases} \\
	\begin{cases}
		P(T_t^* = \frac{\mu^*}{\lambda p^*} \theta w_+(x^*) | h_t \in H^*) &= p^*, \\
		P(T_t^* = 0 | h_t \in H^*) &= 1- p^*,\\
			P(T_t^* = 0 | h_t \not \in H^*) &= 1.
	\end{cases}
	\end{gather*}
	The buyer always naively believes that she will purchase one last time in each period before obtaining the good. The optimal mechanism is generically unique in $H^*$.\footnote{For any probability weighting function $w_+$, there is at most one $\delta$ where the solution is not unique on path (in $H^*$).}
\end{theorem}

In the proposed optimal process, if the buyer ever chooses not to purchase, the seller ends the entire process. Otherwise, in each period, the seller sells a ``loot box'', which provides the buyer the good with probability $x^*$, and charges a binary random price. In each period before receiving the good, the buyer purchases the loot box and naively believes that her future self will stop purchasing. The optimal mechanism is generically unique only on the path: multiple solutions exist because there is some freedom in what happens if the buyer refuses to purchase.

According to Proposition \ref{prop1}, the optimal revenue when the buyer is sophisticated is $\mu^* \theta/\lambda $.  By optimally manipulating the naivety of the buyer, the seller adds an additional multiplier to the revenue:
\begin{gather*}
\max_{x\in[0,1]}  \frac{w_+(x)}{1-(1-x)\delta}.
\end{gather*}

For a better understanding of how the seller exploits the naivety of the buyer, take $\delta =1$, and the factor becomes:
\begin{gather*}
\max_{x\in[0,1]}  \frac{w_+(x)}{x}.
\end{gather*}
When the buyer is sophisticated, she knows the eventual probability of obtaining the product. Any decrease in this probability is an efficiency loss since $w(p)$ is increasing in $p$. Thus, the seller has no choice but to set it to 1. However, when the buyer is naive, she pays in each period for a probability $x^*$ of receiving the product in the optimal mechanism. Even if the product is not delivered in one period, there is no efficiency loss since the seller can sell it in the next period (the buyer in this period does not realize this). Consequently, the seller is able to repeatedly choose the most ``efficient'' probability that maximizes the ratio, without ``wasting'' the remaining probability. If $\delta=1$ and $w_+'(0)=\infty$, the seller can even obtain infinite revenue. However, when $\delta<1$, there is a tradeoff between the efficiency of the probability exchange rate and the delayed payment, and the optimal profit is bounded.

\paragraph{Challenge and Sketch of Proof}
Technically, the main challenge of the proof is the disconnection between the payment and  the ``factual'' probability of delivering the good. In each period, the agent might be willing to pay a large amount of money with a low probability of obtaining the good, under the naive belief that her future decisions can induce a random payment distribution that is favorable under probability weighting. However, because of dynamic inconsistency, the original plan might become unattractive, and the agent might deviate to another plan and voluntarily forgo the initial ``promise''. 

To establish a proper upper bound of the total expected revenue as a function of the  factual probability of delivering the good in each period, I partition the whole time horizon according to the agent's perceived strategy.
Specifically, fix a given $(X_t,T_t)_t$ and an IC-compatible strategy $\sigma$. Recursively define the following stopping time:
\begin{gather*}
	\tau_{-1} =0,\\
	\hat{\tau}_{k+1} = \begin{cases}
		\tau_k,  &\text{ if } \sigma_{\tau_k}^{\tau_k}( h_{\tau_k}) = 1, \\
		\min\{ t > \tau_k|  \sigma^{t}_t(h_t) = 1  \},  &\text{ if } \sigma_{\tau_k}^{\tau_k}( h_{\tau_k}) = 0.
	\end{cases}	\\
	\tau_{k} = 
	\min\{ t > \hat{\tau}_k|  \sigma^{\hat{\tau}_k}_t(h_t) = 0  \}.
\end{gather*}
To illustrate the meaning of the construction, let us consider the initialization.
\begin{itemize}
	\item If the buyer chooses to purchase in period $0=\tau_{-1}$, then $\hat{\tau}_0 = \tau_{-1}$, and $\tau_0$ is the first time the buyer thinks (in period $\hat{\tau}_0$) that she will temporarily stop buying.
	\item If she chooses not to purchase in period $0=\tau_{-1}$, then $\hat{\tau}_0$ is the first time that  she \emph{actually} returns to purchase the good. Furthermore, $\tau_0$ in this case is the first time that the buyer thinks (in period $\hat{\tau}_0$) that she will temporarily stop buying. 
\end{itemize}
By construction, in period $\hat{\tau}_k$, the buyer always thinks that she will keep purchasing during periods $[\hat{\tau}_k,\tau_k)$. Additionally, if $t \not \in [\hat{\tau}_k,\tau_k)$ for all $k$, then $A_t = 0$. That is, the buyer will only buy within some $[\hat{\tau}_k,\tau_k)$. It is unclear whether the buyer actually maintains her plan in period  $\hat{\tau}_k$ and purchases in every period within $[\hat{\tau}_k,\tau_k)$. However, by assuming that she actually insists on purchasing, I can obtain a upper bound on the profit.
\begin{align*}
	\Pi &= \mathbb{E}  \sum_{t=0}^{\infty}  \delta^t     T_t \sigma_t^t(h_t) =  \mathbb{E} \sum_{k=0}^{\infty} \sum_{t= \hat{\tau}_k}^{\tau_k-1}  \delta^t     T_t \sigma_t^t(h_t),\\
	&\leq  \mathbb{E} \sum_{k=0}^{\infty} \sum_{t= \hat{\tau}_k}^{\tau_k-1}  \delta^t     T_t  =  \mathbb{E} \sum_{k=0}^{\infty} \delta^{\hat{\tau}_k} \sum_{t= \hat{\tau}_k}^{\tau_k-1}  \delta^{t-\hat{\tau}_k}     T_t\\
	& = \mathbb{E} \sum_{k=0}^{\infty} \delta^{\hat{\tau}_k} 1_{\tau\geq \hat{\tau}_k} \sum_{t= \hat{\tau}_k}^{\tau_k-1}  \delta^{t-\hat{\tau}_k}     T_t.
\end{align*}

The last line holds because $T_t = 0$ for any $t > \tau$.
Note that by construction, $\tau_k > \hat{\tau}_k \geq \tau_{k-1}$. Therefore, $\hat{\tau}_k \geq k$, and
\begin{align*}
	\Pi &\leq  \mathbb{E} \sum_{k=0}^{\infty} \delta^{k} 1_{\tau\geq \hat{\tau}_k}  \sum_{t= \hat{\tau}_k}^{\tau_k-1}  \delta^{t-\hat{\tau}_k}     T_t,\\
	&= \mathbb{E} \sum_{k=0}^{\infty} \delta^{k}  1_{\tau\geq \hat{\tau}_k}  \mathbb{E} [ \sum_{t= \hat{\tau}_k}^{\tau_k-1}  \delta^{t-\hat{\tau}_k}     T_t | h_{\hat{\tau}_k}, \tau\geq \hat{\tau}_k ]
\end{align*}

Importantly, because continuing to purchase within  $[\hat{\tau}_k,\tau_k)$ is the strategy that the buyer believes that she will use when she faces $h_{\hat{\tau}_k}$ in period $\hat{\tau}_k$, SIR must hold, and I can use Proposition \ref{prop1} to place an upper bound on the discounted transfer:
\begin{align*}
	\mathbb{E} [ \sum_{t= \hat{\tau}_k}^{\tau_k-1}  \delta^{t-\hat{\tau}_k}     T_t | h_{\hat{\tau}_k}, \tau\geq \hat{\tau}_k ] &\leq  \frac{\mu^*}{\lambda}  \int_0^{\infty} w_+ (  P( 1_{\hat{\tau}_k\leq \tau < \tau_k} \delta^{\tau-\hat{\tau}_k} \theta > y|h_{\hat{\tau}_k}, \tau\geq \hat{\tau}_k  ) ) dy,\\
	&\leq \frac{\mu^*}{\lambda}  \int_0^{\infty} w_+ (  P( 1_{\hat{\tau}_k\leq\tau < \tau_k}  \theta > y|h_{\hat{\tau}_k}, \tau\geq \hat{\tau}_k  ) ) dy,\\
	&= \frac{\mu^* \theta}{\lambda}  w_+( P ( \hat{\tau}_k \leq \tau < \tau_k |h_{\hat{\tau}_k}, \tau\geq \hat{\tau}_k) ).
\end{align*}

Therefore, I obtain the following upper bound on the total profit:
\begin{gather*}
	\Pi \leq \mathbb{E} \sum_{k=0}^{\infty} \delta^{k} \frac{\mu^* \theta}{\lambda} 1_{\tau\geq \hat{\tau}_k}   w_+( P ( \hat{\tau}_k \leq \tau < \tau_k  |h_{\hat{\tau}_k},  \tau\geq \hat{\tau}_k ) ).
\end{gather*}

Define scales $\mu_k$, random variable $Y_k$ and (conditional) probability $\bar{F}_k$ as:
\begin{align*}
	Y_k &=  1_{\hat{\tau}_k \leq \tau < \tau_k } , \\
	\bar{F}_k ( \cdot,  h_{\hat{\tau}_k}) &=  P( \cdot |h_{\hat{\tau}_k},  \tau\geq \hat{\tau}_k).
\end{align*}

The upper bound can be expressed as: 
\begin{gather*}
	\Pi \leq \mathbb{E}  \sum_{k=0}^{\infty}  \frac{\mu^* \theta}{\lambda}  \delta^{k} \prod_{i=0}^{k-1}(1-Y_i) \int  w_+(  Y_k  )  d \bar{F}_k.
\end{gather*}

Now, consider a further relaxed problem that maximizes the upper bound over all $(Y_k,\bar{F}_k)$. The maximization problem is  a simple stationary optimization problem, whose optimal value is solved by:
\begin{gather*}
	V =  \max_{F\in \Delta[0,1]} \frac{\mu^* \theta}{\lambda} \int w_+(x) dF + (1-\int x d F) \delta V.
\end{gather*}

Thus far, I have   provided an upper bound on the profit that the seller can collect. This upper bound can be achieved by the proposed loot box mechanism if the buyer continues buying until she obtains the good. The second step is to verify the proposed strategy, namely ``unless having received the good already, always plan to purchase one last time and quit thereafter'' is indeed the buyer's perceived strategy. This part is not as easy as in standard models because of probability weighting, and I relegate it  to the Online Appendix.

\section{Discussion and Extensions}
\subsection{Policy Discussion on the Gacha Games Industry}
As proved in Proposition \ref{prop1}, Theorem \ref{thm3} and Theorem \ref{thm2},  the maximum profit for a stochastic pricing process (under the nonnegativity constraint) is bounded. In contrast, the maximum profit for gambling (without the nonnegativity constraint) is infinite. This distinction provides one rationale for potentially differentiated regulation of stochastic pricing processes, especially if the regulator believes that the probability weighting is the authentic preference of buyers that should be respected.

There is one caveat about the validity of the distinction. Simply preventing the seller from giving the buyer money is insufficient to ensure that the nonnegativity constraint is satisfied because the buyer might be able to resell the product in the second-hand market. What is worse, sellers can secretly repurchase the product in the second-hand market to keep the resale price as high as they want.
There are many possible policies to completely shut down this channel in video games or any virtual industry. One example is to require that any products sold by stochastic processes cannot be transferred  among accounts. In contrast, it is much more difficult to regulate sellers that sell physical goods, such as Pop Mart, which gained a \$6.9 billion  market cap by selling fashion toys via blind boxes. This is because it is almost impossible to ban offline second-hand markets for physical goods.

The above policy suggestion is consistent with developing legal practice. For example, in April of 2018, the Netherlands Gaming Authority conducted a study of 10 unnamed games and concluded that four were in violation of Netherlands laws concerning gambling. Specifically, the study said (via PC Gamer), ``that the content of these loot boxes is determined by chance and that the prizes to be won can be traded outside of the game: the prizes have a market value.'' To sell such items in the Netherlands, a license is required, but given the current laws, no license can be given to game companies, so ``these loot boxes (were) prohibited.'' The loot boxes used in the other games were deemed legal because they lack ``market value.'' According to the study, those loot boxes whose prizes would not be traded constituted a low risk for gambling addiction, being akin to ``small-scale bingo.'' The marketable loot boxes, however, which are banned in the country, ``have integral elements that are similar to slot machines.''\footnote{This paragraph is cited from  \url{https://screenrant.com/lootbox-gambling-microtransactions-illegal-japan-china-belgium-netherlands/}.}

In addition to using regulation to prohibit negative prices, it might also be appropriate for regulators to prevent the exploitation of buyers' naivety regarding their dynamically inconsistent preferences because this kind of exploitation damages social welfare even if the probability weighting is fully respected. 
Recall that the optimal price is ex post bounded for sophisticated buyers but ex post unbounded for naive buyers. Therefore, it might be beneficial to impose some finiteness restriction on prices. For example, the 2016 Chinese regulation  requires game companies to give players an idea of the maximum number (which must be finite) of loot boxes that they would need to  buy to ensure that they obtain a certain item, and it also introduced caps on the number of loot boxes that can be bought in a certain day.

\bigskip

\subsection{Salience Theory}
Apart from prospect theory and RDEU, which feature rank-dependent probability weighting, salience theory is another important and increasingly popular decision-theoretic framework that involves nonlinear probability weighting. As \cite{Shleifer12} noted in their seminar paper, the predictions of salience theory look very similar to those of
prospect theory in certain situations, but in other situations—for instance when small probabilities are not attached to salient payoffs or when lotteries are
correlated—they are very different. The aim of this subsection is to argue that both models have similar predictions in the specific framework of this paper.

To do so, I explore the optimal static mechanism with ex ante IR as in Section 2. I assume that the buyer behaves according to the formulation of \cite{Lanzani22} and assume joint evaluation. The principal designs the joint distribution $F$ of allocation $X$ and nonnegative random payment $T$ that solve
\begin{gather*}
    \max_{X,T\geq 0} \int T ~d F \\
    \int (\theta 1_{X=1} - T ) \sigma(\theta 1_{X=1} - T,0) ~ dF \geq 0,
\end{gather*}
where $\sigma(x,y) > 0$ is the salience function. In this specification, the buyer jointly evaluates the utility from consumption and money and contrast the random distribution $F$ with $0$ (the outsider option). Because of the ordering property of $\sigma$, $\sigma(x,0)$ is decreasing in $x$ if $x<0$ and increasing if $x>0$. Consequently, it is optimal to let $X=1$ for certain.\footnote{Thus, I believe that joint evaluation does not matter, since it will be optimal to set $X=1$ for certain with sensible specifications of multiple variants of salience theory.} In the Online Appendix, I prove 

\begin{proposition}
\label{prop_salience}
Suppose that $x \sigma(x,0)$ is strictly convex when $x>0$ and strictly concave when $x<0$ and that the left derivative of $x \sigma(x,0)$ at $x=0$ is weakly smaller than the right derivative. Then, the unique optimal solution is to deliver the product for certain and charge a binary random price, which is either 0 or a large price that exceeds the value $\theta$.
\end{proposition}
The condition in Proposition \ref{prop_salience} is satisfied for the leading example of \cite{Shleifer12}: \[\sigma(x,y) = \frac{|x-y|}{|x+y|+a}, \quad a>0.\] If the random price can be arbitrarily negative, the optimal solution does not exist since the seller always finds it profitable to introduce arbitrarily rare and negative (salient) payment.

Proposition \ref{prop_salience} states that the optimal static mechanism with ex ante IR takes a similar form as in Proposition \ref{prop1}. Intuitively, since it is always optimal to deliver the product for certain, the only randomness comes from the payment. In this case, the payoff has extreme tail probability if and only if it is salient. Additionally, the correlation does not play a role. I conjecture that the optimal dynamic design remains similar under a proper specification of  salience theory in a dynamic framework.

\subsection{Private Values}
This paper focus on cases where the value of the buyer is known. It is of course natural to seek generalization to private values. As a first step, this subsection generalizes  the static design problem with a sophisticated buyer and ex ante IR. Generalizing the results of dynamic pricing to private value is more challenging because  the buyer can find profitable double deviation in both the report and the dynamic strategy, even if there is no profitable single deviation. Thus, the classical envelopment approach fails. 

Now, suppose that the value of the buyer $\theta \in [\underline{\theta},\overline{\theta}]$ is her private information. The common prior of $\theta$ is $F$ with positive density $f$. The seller can commit to the probability of delivering the good $x(\theta)$ and the distribution of the random price $T(\theta)$ characterized by $P_{\theta}(y)= P(T(\theta)>y)$.
The design problem is:
\begin{gather*}
	\max_{x,P}   \int_{\underline{\theta}}^{\overline{\theta}}  \int_0^{\infty} P(y) dy   ~ f(\theta) d\theta , \\
	s.t. \qquad
	CPT(\theta,x(\theta),P_{\theta}) \geq  CPT(\theta,x(\hat{\theta}),P_{\hat{\theta}}), \\
	v(\theta) = CPT(\theta,x(\theta),P_{\theta}) \geq 0.
\end{gather*}

Although the envelope theorem still holds as in standard framework, the revenue equivalence theorem, which states that the allocation rule and $v(0)$ will determine the revenue, does not hold because the whole distribution of the randomized price matters. However, for any fixed allocation rule, I can characterize the optimal random price according to Proposition \ref{prop1}.
\begin{lemma}
	\label{lemm3}
	For any incentive-compatible mechanism $(x,P)$, the mechanism $(x,P^*)$ is incentive compatible and yields higher expected profit for the seller, where:
	\begin{gather*}
		P^*_{\theta}(y)  = \begin{cases}
			0   \quad  &\text{ if }y \geq  \frac{\mu^*}{\lambda p^*}  \Big( \theta w_+ (x(\theta) ) - v(\theta) \Big),\\
			p^* &\text{ if }y <   \frac{\mu^*}{\lambda p^*}  \Big( \theta w_+ (x(\theta) ) - v(\theta) \Big) ,
		\end{cases}\\
		s.t. \quad 
		v(\theta)= CPT(\theta,x(\theta),P_{\theta}).
	\end{gather*}
\end{lemma}

Lemma \ref{lemm3} relies on the observation that all different types of buyer have the same preference over the random price. The seller can always replace an ``inefficient'' price with the binary random price that leads to the same perceived cost but more revenue.\footnote{However, this argument does not hold with a dynamic mechanism as discussed in Sections 3 and 4, in which case different types with the same report might play different quitting strategies.}

Once focusing on the optimal random price, the revenue equivalence theorem is effectively restored in the sense that the seller can focus on designing the allocation rule $x(\theta)$. I can then follow the Myerson approach to obtain Proposition \ref{prop4} (all proofs are in the Online Appendix). The optimal mechanism is to  post a binary random price for the good.
\begin{proposition}
	\label{prop4}
	For any $\theta^* \in \argmax_{\theta_0}\theta (1-F(\theta ))$, the following post (random) price mechanism $(x^*,P^*)$  is optimal:
	\begin{gather*}
		P^*_{\theta} (y) = 1_{\theta \geq \theta^*} P'(y), \quad x^*(\theta) = 1_{\theta \geq \theta^*}, \\
		\text{where } 	P'(y)  = \begin{cases}
			0   \quad &\text{ if }y \geq \frac{\mu^*\theta^*}{\lambda p^* } ,\\
			p^* &\text{ if }y <   \frac{\mu^*\theta^*}{\lambda p^* } .
		\end{cases} \qquad
	\end{gather*}
\end{proposition}

\subsection{Other Economic Applications}
This paper focuses on the gacha game industry to maximize the tightness of the application. However, the theoretic model can also shed light on other economic applications with different interpretations. Denote the optimal solution in Proposition \ref{prop_expo} as $(\tau_0,T^*(t))$. According to Lemma \ref{lemma4} and Proposition \ref{prop3}, $T^*(\tau_0)$ solves
\begin{gather*}
	max_{T \in \mathcal{T}}  \int_0^{\infty} P(T>y) dy, \\
	s.t. \quad \lambda \int_0^{\infty} w_-(  P(T-s > y) | T> s)  dy \leq \theta, \quad \forall s, ~  s.t. ~ P(T > s) > 0.
\end{gather*}
Next, I discuss two economic applications that can be mapped to this design problem. 
\paragraph{Use Promotion to Incentivize Work}  A firm wants to use promotion (a binary action $p=0,1$) to incentivize a worker to work. The promotion has value $\theta$ to the worker. At each time $t$, the worker can decide whether to shirk or work at a flow cost of $c$ (normalized to 1). Suppose that the firm wants to maximize the expected effort and can design an arbitrary random arrival time of the promotion $T$ that depends on the history of the worker's action. 
\paragraph{Random Answering Time in Queuing System} A queuing system wants to design a random answering time to maximize the expected waiting time for each agent while keeping the agent waiting. The service offered by the system has value $\theta$ to the agent, while waiting in line incurs a flow cost of 1. 

Admittedly, the above specifications are highly stylized in their own settings. For example, it is unrealistic to assume that a promotion can depend on the entire history of a worker's action, and it is difficult for small queuing systems to arbitrarily control the answering time. However, the benchmark solution could shed some light on a more complicated design model with prospect theory agents. In the promotion example, it suggests a tenure system with a bounded length. In the queuing system example, it suggests the use of service-in-random-order (SIRO) with the weight of the waiting agents increasing with their waiting time.

\bibliography{optimal}
\bibliographystyle{chicago}

\section{Appendix}
\begin{proof}[Proof of Lemma \ref{lemm1}]
	First, since $w_-$ is strictly concave in $[0,p_-].$ I know $p/w(p)$ is strictly increasing in $[0,p_-]$. Thus $\mu^* = \sup_{p\in [p_-,1]}w_-(p)$. Because $w_-$ is continuous the optimal $p^*$ exists. Additionally, since $w_-'(1)>1$ I know $\mu^*>1$ and $p^*<1$. 
	
	Then I prove that such $p^*$ is unique. Suppose there is another $\hat{p}$ achieves the bound. We already know it is clear that $p^*,\hat{p} \geq p_-$. Note that by definition the affine function $y(p) = w_-(p^*) + \mu^* (p-p^*) = \mu^* p \leq w(p)$ is a supporting plane for $w_-(p)$ at $w_-(p^*)$. Since $w_-$ is differentiable, the supporting plane at $w_-(p^*)$ is unique and thus every other points in the strictly convex region are strictly above the supporting plane: $w_-(p) > \mu^* p,~ \forall  p^*\neq p \in [p_-,1]$. Therefore $\hat{p} = p^*$.
\end{proof}

\bigskip

\noindent
\begin{proof}[Proof of Proposition \ref{prop1}]
	In the main text I already show that
	\begin{align*}
		\int_0^{\infty} P(y) dy  &=  \int_0^{\infty} w_-(P(y))  \frac{P(y)}{w_-(P(y))} dy\\
		&\leq  \int_0^{\infty} w_-(P(y))  \mu^* dy  \leq \frac{\mu^*\theta}{\lambda }.
	\end{align*}
	Clearly $P^*$ can achieve this upper bound:
	\begin{align*}
		\int_0^{\infty} P^*(y) dy  &=  \int_0^{\frac{\mu^* \theta}{ p^* \lambda}} p^*   dy
		= \frac{\mu^* \theta}{ p^* \lambda} p^*  = \frac{\mu^*\theta}{\lambda }.
	\end{align*}
	$P^*$ is also feasible:
	\begin{gather*}
		\lambda \int_0^{\infty} w_-(P^*(y)) dy = \lambda  \int_0^{\frac{\mu^* \theta}{ p^* \lambda}} w_-(p^*) dy = 
		\lambda \frac{\mu^* \theta}{ p^* \lambda} \frac{p^*}{\mu^*} = \theta.
	\end{gather*}
	Therefore, $P^*$ is an optimal solution. 
	
	As for the uniqueness, from Lemma \ref{lemm1}, the upper bound is achieved only if for a.e. $y>0$, either $P(y) = p^*$  or $P(y) = 0$. Now suppose $P'(y)$ is an optimal solution. Because $P'(y)$ is a decreasing function in $y$, it must take the following form:
	\begin{gather*}
		P'(y) = \begin{cases}
			p^* \quad  &\text{ if } \quad   y < y_0,\\
			0 &\text{ if }  \quad   y\geq y_0.
		\end{cases}
	\end{gather*}
	If $y_0 < {\frac{\mu^* \theta}{ p^* \lambda}}$, the expectation of $P'$ is smaller than the optimal value, which contradicts with the optimality. If $y_0 >  {\frac{\mu^* \theta}{ p^* \lambda}}$, the ex ante IR is violated. In conclusion $P' = P^*$ so the uniqueness is proved.
\end{proof}

\bigskip

\noindent
\begin{proof}[Proof of Lemma \ref{Lemma_eventual}]
	For any $(\tau,T(t))$ such that $P(\tau<\infty)= p < 1$. According to Proposition \ref{prop1}, I know the profit is bounded by $\mu^*\theta/\lambda$, thus for any $\delta$ there exists a time $t_0$ such that
	\begin{gather*}
		\int_{T(t_0)}^{\infty} P(T(\tau)) dy < \delta.
	\end{gather*} 
	Now  construct a new mechanism $(\tau_1,T_1(t))$ such that
	\begin{gather*}
		T_1(t) = T(t) + ( 1 - w_+(p)) \frac{\theta}{\lambda}, \quad \forall t> 0,\\
		\tau_1 = \min \{ \tau , t_0\}.
	\end{gather*}
	In the new payment function $T_1$ I charge one additional lump sum payment at time $0$ and keeps other payment unchanged. By construction I have
	\begin{gather*}
		1_{\tau_1>s}(T(\tau_1) - T(s) )  \leq 1_{\tau > s} ( T(\tau) - T(s)  ),\quad \forall s > 0;\\
		P(\tau_1 < \infty | \tau_1 >s ) = 1 \geq P(\tau < \infty | \tau >s ) ,\quad \forall s > 0.
	\end{gather*}
	Therefore $(\tau_1,T_1(t))$ is DIR compatible for $t> 0$. As for time $0$ DIR condition, I have
	\begin{gather*}
		\theta   w_+(P(\tau_1 < \infty ) - \lambda \int_{0}^{\infty} w_-( P( T_1 (\tau_1) > y   )  ) dy ,\\
		= \theta -  ( 1 - w_+(p))\theta  - \int_0^{\infty} w_-( P( T (\tau_1) > y   )  ) dy,\\
		\geq  \theta w_+(  P(\tau <\infty)  ) - \int_0^{\infty} w_-( P( T (\tau) > y   )  ) dy \geq 0.
	\end{gather*}
	Thus $(\tau_1,T_1(t))$ is indeed DIR compatible. On the other hand, the total revenue for the new mechanism is:
	\begin{gather*}
		\int_0^{\infty} P(T_1(\tau_1 ) > y)  dy = ( 1 - w_+(p)) \frac{\theta}{\lambda} + \int_0^{\infty} P(T(\tau_1 ) > y)  dy,\\
		\geq ( 1 - w_+(p)) \frac{\theta}{\lambda} + \int_0^{T(t_0)} P(T(\tau_1 ) > y)  dy,\\
		\geq ( 1 - w_+(p)) \frac{\theta}{\lambda} + \int_0^{\infty} P(T(\tau ) > y)  dy - \delta.
	\end{gather*}
	Because $p<1$ and $w(p)<1$, I can get strict improvement on total revenue for small $\delta$.
\end{proof}

\bigskip

According to Lemma \ref{Lemma_eventual}, it is without loss of generality to focus on mechanisms where $P(\tau<\infty)=1$. At any time $t$ before the buyer receives the good, she knows that  she will eventually obtain the good for certain if she continues purchasing: $P(\tau<\infty | \tau > t)= 1$. Thus DIR (inequality  \ref{eq_IC}) becomes:
\begin{gather*}
	\lambda \int_{0}^{\infty} w_-( P( T (\tau) - T(s) > y  | \tau >s )  ) dy \leq \theta, \quad   \forall s,~ s.t. ~ P(\tau > s) > 0.
\end{gather*}
In fact, I can further rewrite DIR as:
\begin{lemma}
	\label{lemma3}
	Suppose that $P(\tau<\infty)=1$. Then, $(\tau,T(t))$ is DIR compatible if and only if
	\begin{gather*}
		\lambda \int_{0}^{\infty} w_-( P( T (\tau) - s > y  | T(\tau) >s )  ) dy \leq \theta, \quad   \forall s,~ s.t. ~ P(T(\tau) > s) > 0.
	\end{gather*}
\end{lemma}

Here $T= T(\tau)$ is a random variable representing the random cumulative payment that the buyer has to pay if she continues purchasing until she receives the good. This motivates me to reduce the joint design of $(\tau,T(t))$ to  the design of ``stopping money'' $T(\tau)$.

\noindent
\begin{proof}[Proof of Lemma \ref{lemma3}]
	We already know that $(\tau,T(t))$ is DIR compatible if and only if the following conditions holds: 
	\begin{gather*}
		\lambda \int_{0}^{\infty} w_-( P( T (\tau) - T(s) > y  | \tau >s )  ) dy \leq \theta, \quad   \forall s,~ s.t. ~ P(\tau > s) > 0.
	\end{gather*}
	We want to prove it is equivalent to T-DIR conditions holds:
	\begin{gather*}
		\lambda \int_{0}^{\infty} w_-( P( T (\tau) - s > y  | T(\tau) >s )  ) dy \leq \theta, \quad   \forall s,~ s.t. ~ P(T(\tau) > s) > 0.
	\end{gather*}
	Denote
	\begin{gather*}
		S_1 = \{ s |  P(\tau > s) > 0, ~ T(t)>T(s), ~ \forall t~ s.t.~ P(t>\tau>s)> 0  \},\\
		S_2 =  \{ s |  P(\tau > s) > 0,  ~ \exists t> s ~ s.t.~ T(t)=T(s)\},\\
		T_1  = T(S_1)= \{ T(s) | s \in S_1  \}, \quad  T_2 = \{ s\in \mathbb{R}^+ | P(T(\tau) > s)> 0\} \slash T_1.
	\end{gather*}
	
	I first show that DIR holds for $\forall s,~ s.t. ~ P(\tau > s) > 0$ if and only if DIR holds for $\forall s \in S_1$.		To see this, denote for $s\in S_2$ 
	\begin{gather*}
		I(s) = \{  t\geq s |  T(t) = T(s) \}.
	\end{gather*}
	If there is no $t > s$ such that $t\in S_1$, which means there is no further payment, DIR holds trivially. Otherwise, because $T$ is increasing left-continuous function, $I(s) = [s,t_0]$, $t_0 \in S_1$ and 
	\begin{align*}
		P( T (\tau) - T(s) > y  | \tau >s )  &= \frac{P( T (\tau) - T(t_0) > y)}{P(\tau > s)} \\
		&\leq \frac{P( T (\tau) - T(t_0) > y)}{P(\tau > t_0)} = P( T (\tau) - T(t_0) > y  | \tau >t_0 ).
	\end{align*}
	Thus, DIR holds in $s$ if it holds in $t_0$. This means if DIR holds for any $s\in S_1$, it also holds for any $s\in S_2$.
	
	Second, I show that T-DIR holds for  $\forall s,~ s.t. ~ P(T(\tau) > s) > 0$ if and only if T-DIR holds for any $s \in T_1$. It is easy to see $T(S_2) \subseteq T(S_1)$ so $T_1 = supp T(\tau)$. Since $0 \in T_1$ and
	\begin{gather*}
		\{ s\in \mathbb{R}^+ | P(T(\tau) > s)> 0\} \subseteq [0,  \max suppT(\tau)],
	\end{gather*}
	For any $s$ such that $ P(T(\tau) > s)> 0$ but $s \not \in T_1$, there exists a $s_0 \in T_1$ such that $s_0<s$ and $(s_0,s) \cap T_1 =\emptyset$, namely, $P(s_0< T(\tau) \leq s) = 0$. Consequently 
	\begin{gather*}
		P(T(\tau) - s> y | T(\tau) > s ) = P( T(\tau) - s >y |  T(\tau) > s_0 ) \leq P( T(\tau) - s_0 >y |  T(\tau) > s_0 )
	\end{gather*}
	Thus, T-DIR holds at $s$ if it holds at $s_0 \in T_1$.

	Third, it is immediate to see the equivalence between DIR on $S_1$ and T-DIR on $T_1$ because $T$ is a one to one mapping between $s\in S_1$ and $T(s) \in T_1$ such that
	\begin{gather*}
		\{T(\tau) >T(s) \}     =  \{\tau >s \} .  
	\end{gather*}
	Combining the three steps, the lemma is proved.
\end{proof}

\bigskip

\noindent
\begin{proof}[Proof of Lemma \ref{lemma4}]
	We prove by contradiction. Suppose there exists $(\tau_1,T_1(t))$ that is DIR compatible and leads to more profit than $(\tau,T(t))$. According to Lemma \ref{lemma3}, $T_1(\tau_1)$ is feasible in the auxiliary   problem. However, this contradicts with the optimality of $T(\tau)$.
\end{proof}

\bigskip

\noindent
\begin{proof}[Proof of Proposition \ref{prop3}]
	I only need to prove for any non-negative random variable $T \in \mathcal{T}$, there exists $T(t)$ such that 
	\begin{gather*}
		P( T \leq t )  =  P( T(\tau_0) \leq t) .
	\end{gather*}
	According to  Lemma \ref{lemma3},  $(\tau_0,T)$ satisfies DIR.
	
	Since $T \in \mathcal{T}$, by definition I can write the cumulative function as:
	\begin{gather*}
		P(T \leq y) = \int_0^y \rho(s) ds + \sum_i 1_{y_i \leq y}  Y_i.
	\end{gather*}
	Define $P(y)$ and $T(t)$ as:
	\begin{align*}
		P(y) &= P(T < y) = \int_0^y \rho(s) ds + \sum_i 1_{y_i < y}  Y_i, \quad \forall y \in \mathbb{R}^+,\\
		T(t) &= \max \{  y \in \mathbb{R}^+  | 1- P(y) \geq e^{-t}  \}  \quad \forall t>0.
	\end{align*}
	Note that $P(0)=0$ and $P(y)$ is left-continuous and increasing, so $T(t)$ is well-defined,  increasing and left-continuous. In addition,
	\begin{gather*}
		T(t) \geq y   \Leftrightarrow   e^{-t} \leq   1- P(y).
	\end{gather*}
	Thus,
	\begin{gather*}
		P(T(\tau)  \geq y ) = P( e^{-\tau} \leq   1- P(y)  ) = 1- P(y) = P(T \geq y).
	\end{gather*}
	The proof that $T$ admits regular a decomposition is not of great interest economically and I relegate it to the Online Appendix.
	
\end{proof}

\bigskip

\bigskip

\noindent
\begin{proof}[Proof of  Proposition \ref{prop_expo}]
	We first use the following lemma to ensures the existence and uniqueness of $f$ described in the proposition.
	\begin{lemma}
		\label{lemma7}
		For any constant $T_1 > 0$, there exists a unique $f$ in $L_1[0,T_1]$ such that
		\begin{gather*}
			w_-(e^{-T_1 + t})  +  \int_t^{T_1} w_- (e^{-s+t}  ) f(s) ds  = 1, \quad \forall t \in [0,T_1].
		\end{gather*}
	\end{lemma}

	The proof of Lemma \ref{lemma7} is in the Online Appendix. Now  I continue the analysis of proof sketch in the main text. Recall that
	with the help of Lemma \ref{lemma5}, I can rewrite our exponential design problem as
	\begin{gather*}
		V= \max_{f,T_i,t_i}   \sum_i T_i e^{-t_i} + \int_0^{\infty}e^{-t}  f(t) dy  ,\\
		s.t.  \quad   \lambda \sum_{i,t_i \geq s} T_i w_-( e^{- t_i + s} ) + \lambda 	\int_{s}^{\infty} w_-( e^{-t +s }) f(t) dt - \theta \leq 0, \quad \forall s\geq0.
	\end{gather*}
	To solve it, I  use  the duality approach. Consider the following relaxed problem:
	\begin{gather*}
		V(\gamma) = \max_{f,T_i,t_i} \quad  \sum_i T_i e^{-t_i} + \int_0^{\infty}e^{-t}  f(t) dy -  \Big(   \sum_i T_i w_-(e^{-t_i}) + \int_0^{\infty} w_-(e^{-t})  f(t) dy - \frac{\theta}{\lambda} \Big) \\
		- \int_0^{\infty} \gamma(s)  \Big(  \sum_{i,t_i \geq s} T_i w_-( e^{- t_i + s} ) +  	\int_{s}^{\infty} w_-( e^{-t +s }) f(t) dt - \frac{\theta}{\lambda}  \Big) ds.
	\end{gather*}
	
	Suppose the Lagrange density\footnote{Note that besides the Lagrange density $\gamma(s)$ whose value is to be determined later, I also explicitly add one unit mass of Lagrange measure on the DIR condition for time $0$.} $\gamma(s)$ is non-negative. Then any feasible $(f,T_i,t_i)$ in the exponential design problem  has higher objective value  than in the relaxed problem. Thus, the optimal value of the relaxed problem (if exits) is weakly higher than the optimal value of the exponential design problem:  $V(\lambda) \geq V$.
	
	Now using integration by parts, I can transform the  relaxed problem as
	\begin{gather*}
		V(\gamma) = \max_{f,T_i,t_i} \quad \int_0^{\infty} f(t) \Big( e^{-t} - w_-(e^{-t}) - \int_0^t  w_-(e^{-t+s}) \gamma(s) ds  \Big) dt + \\
		\sum_i T_i  \Big(  e^{-t_i} - w_-(e^{-t_i})  - \int_0^{t_i}  w_-(e^{-t+s}) \gamma(s) ds   \Big) +V_0 \\
		\text{where} \quad V_0 = \frac{\theta}{\lambda} ( 1 + \int_0^{\infty} \gamma(s) ds ).
	\end{gather*}
	
	To proceed, I use the Lagrange density $\gamma(t)$ with the following property:
	\begin{lemma}
		\label{lemma6}
		There exists $\gamma(t)$ and $t_0>0$ such that
		\begin{itemize}
			\item  $\gamma(t) > 0$ for $t \in [0,t_0)$, and $\gamma(t) = 0$ for $t \in [t_0,\infty)$.
			\item $e^{-t} - w_-(e^{-t}) - \int_0^t  w_-(e^{-t+s}) \gamma(s) ds = 0$ for $t\in [0,t_0]$;
			\item  $e^{-t} - w_-(e^{-t}) - \int_0^t  w_-(e^{-t+s}) \gamma(s) ds < 0$ for $t\in (t_0,\infty)$.
		\end{itemize}
	\end{lemma}
	The proof of Lemma \ref{lemma6} is in the Online Appendix. 
	Crucially, the existence of such Lagrange density relies on the curvature of the probability weighting function: this is where I need the positive third derivatives. 
	
	With the Lagrange density $\gamma(t)$ constructed in Lemma \ref{lemma6}, the optimal value of the relaxed problem $V(\gamma) = V_0$. Thus, I know $V\leq V(\gamma) = V_0$. On the other hand, consider the cumulative payment function described in Proposition \ref{prop_expo}: 
	\begin{gather*}
		T^*(t) =\frac{\theta}{\lambda}  \Big(     1_{t>t_0}   + \int_0^{\min \{ t,t_0 \}}  f(s) ds  \Big).
	\end{gather*}
	$T^*(t)$ has value $V_0$ in the relaxed problem because it stays constant in $(t_0,\infty]$ (where $\gamma(s)<0$). It also has value $V_0$ in the original exponential design problem because its DIR is binding for $t\in[0,t_0]$,  which makes sure the complementary slackness holds:
	\begin{gather*}
		\gamma(s) \Big( 1_{t_0\geq s} w_-( e^{- T + s} ) \frac{\theta}{\lambda}  +  	\int_{s}^{t_0} w_-( e^{-t +s }) \frac{\theta}{\lambda}  w_-'(e^{-t_0+t})e^{-t_0+t}  dt - \frac{\theta}{\lambda}  \Big) = 0 \quad \forall s \geq 0.
	\end{gather*}
	Since $T^*(t)$ is DIR compatible by construction, it is indeed the optimal cumulative payment function.

	Finally, I prove the uniqueness. Suppose 
	\begin{gather*}
		T^1(t) = \int_0^t f(s) ds + \sum_{i=0}^{\infty} 1_{t_i<t} T_i
	\end{gather*} is also an optimal cumulative payment function. Since I already know $V = V(\lambda) = V_0$, $T^1(t)$ must have value $V_0$ in the relaxed problem. This means  all the jump points $t_i \leq t_0$ and $f(s)$ is almost everywhere 0 in $[t_0,\infty]$, and consequently 
	\begin{gather*}
		T^1(t) = T^1(t_0+) \quad \forall t > t_0.
	\end{gather*}
	
	Because $T^1(t)$ has the same value in the exponential design problem and  the relaxed problem, Complementary Slackness condition must hold almost everywhere.
	This joint with $\gamma(s)>0$ for $s\in[0,t_0]$ implies DIR must be binding almost everywhere for $s \in [0,t_0]$. In addition, note that the expression of DIR:
	\begin{gather*}
		h(s) = \sum_{i,t_i \geq s} T_i w_-( e^{- t_i + s} ) +  	\int_{s}^{\infty} w_-( e^{-t +s }) f(t) dt - \frac{\theta}{\lambda}  
	\end{gather*}
	is left continuous in $s$. If DIR is not binding at some time $s_0\in(0,t_0]$, then DIR is not binding for all $s \in (s_0-\delta,s_0]$. This is a contradiction.  
	
	Now that DIR is binding for any $t \in(0,t_0]$, it must also be binding at time $0$. This is because otherwise the seller can charge an additional lump-sum payment at time 0, which contradicts with the optimality of $T^1$.  In conclusion, $T^1(t) = T^*(t)$.
	
\end{proof}

\bigskip

\noindent
\begin{proof}[Proof of Theorem \ref{thm3}]
	From Proposition \ref{prop_expo} I know that $T^*(t)$ is the optimal cumulative payment function for the exponential stopping time $\tau_0$. Next, I claim $T^*(\tau_0)$ solves the  the auxiliary   problem:
	\begin{gather*}
		T^*(\tau_0) \in  \argmax_{T \in \mathcal{T}}  \int_0^{\infty} P(T>y) dy, \\
		s.t. \quad \lambda \int_0^{\infty} w_-(  P(T-s > y) | T> s)  dy \leq \theta, \quad \forall s, ~  s.t. ~ P(T > s) > 0,
	\end{gather*}
	
	Suppose for the purpose of contradiction that there exists $T_1$ that yields a higher expected profit than $T^*(\tau_0)$. According to Proposition \ref{prop3}, there exists $T_1(t)$ such that $(\tau_0,T_1(t))$ is DIR compatible in the original design problem and yields the same expected  profit as $T_1$, which is higher  than $(\tau_0,T^*(t))$. This contradicts  the optimality of $T^*(t)$.
	
	Since $T^*(\tau_0)$ solves the auxiliary   problem, according to Lemma \ref{lemma4}, $(\tau_0,T^*(t))$ is the optimal dynamic pricing process in the original problem. 
	
	For the uniqueness part, suppose $(\tau,T)$ is a optimal dynamic pricing process. According to Lemma \ref{lemma3}, $T(\tau)$ is feasible in the auxiliary   problem. Then according to Proposition \ref{prop3}, there exists a cumulative payment function $T_2(t)$ such that $(\tau_0,T_2)$ is DIR compatible and
	\begin{gather*}
		P(T(\tau) \leq y) =P(T_2(\tau_0) \leq y), \quad \forall y \in \mathbb{R}^+.
	\end{gather*}
	Since $T^*(\tau_0)$ solves the  the auxiliary   problem, so does $T(\tau)$  and
	$T_2(\tau_0)$ (they have the same expected profit). Consequently, $(\tau_0,T_2)$ is an optimal dynamic pricing process. The uniqueness part of Proposition \ref{prop_expo} implies $T_2 = T^*$ so 
	\begin{gather*}
		P(T(\tau) \leq y) =P(T_2(\tau_0) \leq y)  =P(T^*(\tau_0) \leq y), \quad \forall y \in \mathbb{R}^+.
	\end{gather*}
	
\end{proof}

\bigskip

\noindent
\begin{proof}[Proof of Proposition \ref{robust_RankDependent}]
	To prove the proposition, first note that in both design problems (for Proposition \ref{prop1} and Theorem \ref{thm3}) it is optimal to set $X=1$. This is because the seller can always increase the payment, as he increases the probability of delivering, to keep the ex-post payoff unchanged. Then, the payoff of a rank-dependent utility buyer is
	\begin{align*}
		RDEU(1,T,\theta) &= \int_0^{\infty} w(P( \theta  - T > y  )) dy+ \int_{-\infty}^0 (w(P( \theta  -T  >y  )) -1 )dy\\
		&=  \int_0^{\theta} w(P( \theta  - T > y  )) dy+ \int_{-\infty}^0 (w(P( \theta  -T  >y  )) -1 )dy \\
		&=\int_0^{\theta} w( 1 -P(   T \geq \theta - y  )) dy+ \int_{-\infty}^0 (w( 1 -P(   T \geq \theta - y  ) ) -1 )dy \\
		&= \int_0^{\theta}  w( 1 -P(   T \geq y  )) dy + \int_{\theta}^{\infty} (w( 1 -P(   T \geq y  ) ) -1 )dy \\
		&= \theta - \int_0^{\theta} w_-( P(T \geq y) ) dy -  \int_{\theta}^{\infty} w_-( P(T \geq y) ) dy \\
		&= \theta - \int_0^{\infty} w_-( P(T > y) ) dy = CPT(1,T,\theta).
	\end{align*}
	The last equation holds because there is at most countable many $y$ such that $P(T\geq y) \neq P(T>y)$. Now because the payoff of the agent for any random payment $T$ is the same, all results go through.
\end{proof}
\bigskip

\noindent
\begin{proof}[Proof of Proposition \ref{prop5}]
	The proof is similar to the proof of  Proposition \ref{robust_RankDependent} and I relegate it to Online Appendix.
\end{proof}

\newpage 

\section{Online Appendix (For online publication only)}

\subsection{Remaining Proof of Proposition \ref{prop3}}
\begin{proof}
	In the appendix, I construct $T(t)$ that is used for replication. The remaining task is to prove $T$ admits a regular decomposition. Denote
	\begin{align*}
		\mathbb{I} &= \{  (a,b) | ~ P(y) \text{  is absolute continuous and strictly increasing in} (a,b)  \},\\
		I &= \cup_{(a,b)\in \mathbb{I}} (a,b);\\ 
		\mathbb{J} &= \{   (c,d) |~  P(y) \text{ is constant over } (c,d)    \}, \\
		J &= \cup_{(c,d) \in \mathbb{J}} (c,d).
	\end{align*}	
	$I$ and $J$ can be expressed as a countable union of disjoint open intervals. and I have the following partition of $\mathbb{R}^+$:
	\begin{gather*}
		\mathbb{R}^+  = I \cup J \cup A.
	\end{gather*}
	Because $P$ admits regular decomposition, The point in $A$ is either the end point of those disjoint countable intervals in $I$ and $J$, or one of the countable jump points of $P$. Therefore, $A$ is a countable set.
	
	Now I want to prove $T(t)$ admits regular decomposition by proving $T(t)$ is absolute continuous over $R^+$ except a countable set. To do so, I decompose the $R^+$ as
	\begin{gather}
		T^{-1}(I) \cup T^{-1}(J) \cup T^{-1}(A).
	\end{gather}
	First, for any $t\in R^+$ such that $T(t) \in I$, the following equation holds in its neighbor,
	\begin{gather*}
		1-P(T(t)) = e^{-t} \implies T(t) = P^{-1} ( 1- e^{-t} ).
	\end{gather*}
	Since $P$ is absolute continuous and strictly increasing in $I$, $P^{-1}$ is absolute continuous in the neighborhood and so is $T(t)$. 
	
	Second, denote $J= \cup_i (c_i,d_i)_i$ where $(c_i,d_i)$ is disjoint with other. Now I claim that, there is no $t$ such that $T(t) \in (c_i,d_i)$. Thus, the set $T^{-1}(J) = \{ t \in R^+| T(t) \in J \}$ is empty. To prove the claim, recall $P$ is constant over $(c_i,d_i)$ so $	P(c_i)  = P(d_i) = p_i.$ 
	Now suppose there exists $i$ and $t$ such that $T(t) \in (c_i,d_i)$. By definition of $T(t)$, for any $\varepsilon>0$, 
	\begin{align*}
		1-P(T(t)) &\geq e^{-t}, \\
		1-P(T(t)+ \varepsilon) & < e^{-t}.
	\end{align*}
	From the monotonicity of $P$, for sufficiently small $\varepsilon$,
	\begin{gather*}
		1-P(c_i) \geq 1-P(T(t)) \geq e^{-t} > 	1-P(T(t)+ \varepsilon) \geq 1-P(d_i) = 1-P(c_i).
	\end{gather*}
	This is a contradiction.

	Third, I consider $T^{-1}(A)$. Because $T(t)$ is increasing,  for any $y \in T(R^+)$, $T^{-1}(y) = \{ t|   T(t) = y  \}$ is either one point or an interval.  If $T^{-1}(y)$ is an interval then $T(t)$ is absolute continuous in  it.   Because $A$ is countable, the remaining set is at most countable.

	In conclusion, $T(t)$ is absolute continuous except a countable set. Therefore, $T$ has regular decomposition.
\end{proof}

\subsection{Proof of Lemma \ref{lemma5}, \ref{lemma7} and  \ref{lemma6}}
\noindent
\begin{proof}[Proof of Lemma \ref{lemma5}]
	First note that by defining $T_s(t) = T(t+s)- T(s)$ and $t^s_i = t_i -s$, I can normalize the equation in Lemma \ref{lemma5}  as
	\begin{gather*}
		\int_{0}^{\infty} w_-( e^{-T_s^{-1}( y )}    ) dy  =  \sum_{i,t_i^s \geq 0} T_i w_-( e^{- t_i^s} ) + 	\int_{0}^{\infty} w_-( e^{-t }) T'_s(t) dt , \quad \forall s\geq 0.
	\end{gather*}
	Therefore, it is without loss of generality to prove the lemma only for $s= 0$.
	
	To be more reader friendly, let me first review the definition of $T^{-1}$.
	\begin{gather*}
		T^{-1}(x) = \sup \{ t \in \mathbb{R}^+ |T(t) \leq x   \}.
	\end{gather*}
	Since $T(t)$ jumps from $T(t_i)$ to $T(t_i+) = T(t_i) + T_i$ at time $t_i$,  $T^{-1} (x) = t_i$ for $x \in [T(t_i),T(t_i) +T_i)$. Denote $\mathbb{I} = \mathbb{R}^+ /\ (\cup_i [T(t_i),T(t_i)+T_i)  )$. $\mathbb{I}$ is a countable union of  interval in which $T$ is absolutely continuous.
	\begin{align*}
		\int_{0}^{\infty} w_-( e^{-T^{-1}( y ) }    ) dy  
		&=  \sum_{i} \int_{T(t_i),T(t_i)+T_i} w_-( e^{-T^{-1}( y ) }    ) dy +  \int_{\mathbb{I}} w_-( e^{-T^{-1}( y ) }    ) dy \\
		&=  \sum_i T_i w_i ( e^{-t_i}) +  \int_{\mathbb{I}} w_-( e^{-T^{-1}( y ) }    ) dy.
	\end{align*}
	Within $\mathbb{I}$, I need to worry about intervals in which $T(t)$ stays constant. Define
	\begin{align*}
		\mathcal{M} &=  \{  [s_1,s_2]   ~ | ~ T(s_1) = T(s_2)     \},\\
		\mathbb{J}  &=  \cup_{J \in \mathcal{M}} ~ J.
	\end{align*}
	The interior of $\mathbb{J}$ is clearly a countable union of disjoint interval: $\mathbb{J}^{\circ} = \cup_i (a_i,b_i) $. We will focus on the interior set here and in the following, because I do not need to worry about (countable) boundary points when doing integration with respect to Lebesgue measure.  
	
	By construction $T(a_i+) = T(b_i) = y_i$, and  $(\mathbb{R}^+/ \mathbb{J})^{\circ} = (0,a_1) \cup_i (b_i,a_{i+1}) \cup (\sup b_i, \infty)$ is a countable union of disjoint interval, which  I simply denote as $\cup_i (b_i,a_{i+1})$. These two facts imply that up to countable points,
	\begin{gather*}
		T(\mathbb{R}^+) = \cup_i (T(b_i),T(a_{i+1}) ).
	\end{gather*}

	Next, I take a close look into the interval $(T(b_i),T(a_{i+1}) )$ and discuss the possible confusion caused by jump points. The set $((b_i,a_{i+1}) / \{ t_k \}_k)^{\circ}$ is also a countable union of intervals (recall that $\{t_k\}_k$ is the countable set of jump points of $T$). Denote it as $\cup_j (c^i_j,d^i_j)$, I have
	\begin{gather*}
		(	(T(b_i), T(a_{i+1})) \cap \mathbb{I} )^{\circ} = \cup_j ( T(c^i_j) , T(d^i_{j})  ).
	\end{gather*}
	
	By construction $T$ is absolute continuous in $(c^i_j,d^i_j)$, its derivative $f$ is almost everywhere positive. Thus, $T$ is strictly increasing and $T^{-1}$ coincides with its inverse. $T^{-1}$ is also absolute continuous and has positive derivative almost everywhere. In addition, $T^{-1}(T(c^i_j)) = c^i_j$ and $T^{-1}(T(d^i_j)) = d^i_j$. Therefore 
	\begin{align*}
		\int_{(T(b_i),T(a_{i+1})) \cap \mathbb{I}} w_-( e^{-T^{-1}( y ) }    ) dy 
		&= \sum_j \int_{T(c_j^i)}^{T(d_j^i)} w_-( e^{-T^{-1}( y ) }    ) dy\\
		&= \sum_j \int_{c^i_j}^{d^i_j} w_-( e^{-t  }) T'(t) dt
		=    \sum_j \int_{c^i_j}^{d^i_j} w_-( e^{-t  }) f(t) dt \\
		&= \int_{b_i}^{a_{i+1}} w_-( e^{-t  }) f(t) dt.
	\end{align*}
	Now since $f(t) = T'(t)$ is almost everywhere $0$ in $\mathbb{J}$, I can conclude:
	\begin{align*}
		\int_{0}^{\infty} w_-( e^{-t }) f(t) dt 
		&= \int_{\mathbb{J}}  w_-( e^{-t  }) f(t) dt +  \sum_i  \int_{b_i}^{a_{i+1}} w_-( e^{-t  }) T'(t) dt\\
		&=  \sum_i  \int_{b_i}^{a_{i+1}} w_-( e^{-t }) T'(t) dt\\
		&= \sum_i \int_{(T(b_i),T(a_{i+1})) \cap \mathbb{I}} w_-( e^{-T^{-1}( y ) }    ) dy \\
		&= \int_{\mathbb{I}}  w_-( e^{-T^{-1}( y ) }    ) dy .
	\end{align*}
	In summary, I have proved:
	\begin{align*}
		\int_{0}^{\infty}  w_-( e^{-T^{-1}( y ) }    ) dy  &=  \sum_{i} T_i w_-( e^{- t_i} ) + \int_{\mathbb{I}}  w_-( e^{-T^{-1}( y ) }    ) dy \\
		&= \sum_{i} T_i w_-( e^{- t_i} ) + 	\int_{0}^{\infty} w_-( e^{-t  }) f(t) dt .
	\end{align*}
\end{proof}

\bigskip

\noindent
\begin{proof}[Proof of Lemma \ref{lemma7}]
	It is sufficient to prove the differential form, namely, to prove there exists a unique $f$ in $L_1[0,T_1]$ such that
	\begin{gather*}
		f(t) = w'_-(e^{-T_1+t}) e^{-T_1+t} + \int_t^{T_1} w_-'(e^{-s+t}) e^{-s+t} f(s) ds, \quad \forall t \in [0,T_1].
	\end{gather*}
	Consider the following self-mapping on $L_1[0,T_1]$, $\Gamma$:
	\begin{gather*}
		\Gamma(f) (t) =  w'_-(e^{-T_1+t}) e^{-T_1+t} + \int_t^{T_1} w_-'(e^{-s+t}) e^{-s+t} f(s) ds, \quad \forall t \in [0,T_1].
	\end{gather*}
	First I show that this is indeed a self-mapping. Integration by parts I have
	\begin{align*}
		\int_0^{T_1} |	\Gamma(f) (t) |  dt &= 	\int_0^{T_1} \Big| w'_-(e^{-T_1+t}) e^{-T_1+t} + \int_t^{T_1} w_-'(e^{-s+t}) e^{-s+t} f(s) ds \Big| dt \\
		&\leq 1 - w_-(e^{-T_1}) + \int_0^{T_1}\int_0^{s} w_-'(e^{-s+t}) e^{-s+t} | f(s) | dt ds   \\
		&=  1 - w_-(e^{-T_1}) +  \int_0^{T_1} ( 1- w_-(e^{-s}))  |f(s)| ds < \infty.
	\end{align*}
	Next I prove $\Gamma$ is a contracting mapping, using the same integration by parts one can show
	\begin{align*}
		\|\Gamma(f_1) - \Gamma(f_2) \|_1 &=	\int_0^{T_1} |	\Gamma(f_1) (t) - \Gamma(f_2) (t) |  dt\\
		&\leq \int_0^{T_1} ( 1- w_-(e^{-s}))  |f_1(s) - f_2(s)| ds \\
		&\leq (1- w_-(e^{-T_1})) \| f_1 - f_2\| _1.
	\end{align*}
	Thus I can apply contracting mapping theorem to prove the existence and uniqueness of $f$, which is the fixed point of $\Gamma$.
\end{proof}

\bigskip

\noindent
\begin{proof}[Proof of Lemma \ref{lemma6} ]
	Recall the model assumes $w_-'(0),w_-'(1)>1$ and $w_-'''(p)>0$. This implies $w_-$ is strictly concave in $[0,p_-]$, strictly convex in $[p_-,1]$. Additionally, it is easy to see there exists a unique $p_1\in (0,1)$ such that $w_-(p_1) = p_1$, $w_-(p) \leq p$ in $[0,p_1]$, and $w_-(p) \geq p$ in $[p_1,1]$. 
	
	I first prove that there exists a function $g_0$ defined on $[0,t_1] = [0, - \log p_1]$ such that
	\begin{gather*}
		e^{-t} - w_-(e^{-t}) - \int_0^t  w_-(e^{-t+s}) g_0(s) ds = 0, \quad \forall t\in [0,t_1].
	\end{gather*}
	It is sufficient to prove with its differential form, namely, to prove there exists $g_0(t)$ such that
	\begin{gather*}
		g_0(t) = 	-e^{-t} + w'_-(e^{-t}) e^{-t} - \int_0^t w_-'(e^{-t+s})  e^{-t+s} g_0(s) ds, \quad \forall t\in [0,t_1].
	\end{gather*}
	Consider the functional space $L_1[0,t_1]$ (which is a Banach space).  Define the following self-mapping $\Gamma$ on  $L_1[0,t_1]$ :
	\begin{gather*}
		\Gamma(g) (t) = -e^{-t} + w'_-(e^{-t}) e^{-t} - \int_0^t w_-'(e^{-t+s})  e^{-t+s} g(s) ds.
	\end{gather*}
	To ensure the my statement is well-defined, I verify that $\Gamma(g) \in L_1[0,t_1]$. 
	\begin{gather*}
		\int_0^{t_1} |\Gamma(g) (t)| dt = \int_0^{t_1} \Big| -e^{-t} + w'_-(e^{-t}) e^{-t} - \int_0^t w_-'(e^{-t+s})  e^{-t+s} g(s) ds \Big| dt\\
		\leq 2 -e^{-t_1} - w_-(e^{-t_1}) + \int_0^{t_1} |g(s) |\int_s^{t_1} w_-'(e^{-t+s}) e^{-t+s} dt  ds\\
		=    2 -e^{-t_1} - w_-(e^{-t_1}) + \int_0^{t_1} |g(s)| ( 1- w_-(e^{-t_1 + s})  )  ds < \infty.
	\end{gather*} 
	Next, I prove that $\Gamma$ is a contracting mapping. For any $g_1, g_2 \in L_1[0,t_1]$
	\begin{align*}
		\| \Gamma(g_1) - \Gamma(g_2)\|_1 &=  \int_0^{t_1}  | \Gamma(g_1)(t) - \Gamma(g_2) (t) |  dt\\  
		&\leq \int_0^{t_1}  \int_0^{t} w_-'(e^{-t+s}   )e^{-t+s}   |g_1(s)  -g_2(s)  | ds  dt\\
		&=  \int_0^{t_1} |g_1(s)  -g_2(s)  |  \int_s^{t_1} w_-'(e^{-t+s}   )e^{-t+s}      dt ds \\
		&= \int_{0}^{t_1} ( 1- w_-(e^{-t_1 + s})  )     | g_1(s)  -g_2(s)  | ds\\
		&\leq( 1- w_-(e^{-t_1})  )    \int_{0}^{t_1}    | g_1(s)  -g_2(s)  | ds =  ( 1- w_-(e^{-t_1})  )    \| g_1- g_2 \|_1.
	\end{align*}
	Therefore, according to contracting mapping theorem (or Banach fixed point theorem), there exists a unique fixed point $g_0$ of $\Gamma$. This completes my first step.
	
	Second, define $t_0$ as:
	\begin{gather*}
		t_0 = \min \{  t \in [0,t_1]  |   g_0(t) \leq 0 \}.
	\end{gather*} 
	$t_0$ is well-defined, because $g_0$ is clearly continuous in $[0,t_1]$ and if $g_0(t) > 0$ for all $t \in [0,t_1]$, then the following contradiction arises:
	\begin{gather*}
		0=	e^{-t_1} - w_-(e^{-t_1}) - \int_0^{t_1}  w_-(e^{-t_1+s}) g_0(s) ds = - \int_0^{t_1}  w_-(e^{-t_1+s}) g_0(s) ds <0.
	\end{gather*}
	The first equality comes from the definition of $g_0$, and the second equality comes from the definition of $t_1$. Additionally, $t_0>0$ because it is easy to see $g_0(0) = w_-'(1)-1 > 0$.
	
	Now I construct the Lagrange density $\gamma(t)$ as:
	\begin{gather*}
		\gamma(t) = \begin{cases}
			g_0(t)  \quad &\forall t \in [0,t_0],\\
			0 \quad  &\forall t \in [t_0,\infty).  
		\end{cases}
	\end{gather*}
	What remains to be proved is
	\begin{align*}
		h(t) =& e^{-t} - w_-(e^{-t}) - \int_0^t  w_-(e^{-t+s}) \gamma(s) ds \\
		=&  e^{-t} - w_-(e^{-t}) - \int_0^{t_0}  w_-(e^{-t+s}) \gamma(s) ds	< 0, \quad \forall t \in (t_0,\infty).
	\end{align*}
	First note that since $w_-(p) > p$ for any $p \in (0,p_1)$,
	\begin{gather*}
		h(t) < - \int_0^{t_0}  w_-(e^{-t+s}) \gamma(s) ds	 < 0, \quad \forall t \geq t_1.
	\end{gather*}
	Thus, I only need to worry about $t\in (t_0,t_1)$. By construction I know $h(t_0) = 0$ and $h'(t_0) \leq 0$. Define
	\begin{align*}
		h_1(t) &= h'(t) e^t= -1 + w_-'(e^{-t}) + \int_0^{t_0} w_-'(e^{-t+s}) \gamma(s) e^s
		ds,\\
		h_1'(t) &= -w_-''(e^{-t}) e^{-t}  -  \int_0^{t_0} w_-''(e^{-t+s}) e^{-t+2s} \gamma(s) ds,\\
		h_2(t) &= e^t h_1'(t) =  -w_-''(e^{-t})  -  \int_0^{t_0} w_-''(e^{-t+s}) e^{2s} \gamma(s) ds,\\
		h_2'(t) &=  w_-'''(e^{-t})e^{-t}  +  \int_0^{t_0} w_-'''(e^{-t+s}) e^{-t+3s} \gamma(s) ds.
	\end{align*}
	
	Because $w_-'''(p) > 0$, $h_2'(t)>0$ and $h_2$ is increasing in $t$. Thus, $h_2$ and $h_1'$ can be  always positive, always negative, or first negative then positive over $(t_0,t_1)$. This implies $h_1$ and $h'$ can be always increasing, always decreasing, or first decreasing then increasing over $(t_0,t_1)$. This further implies $h$ can be always convex, always concave, or first concave then convex over $(t_0,t_1)$.  
	This final statement, jointly with the fact that $h(t_0)=0$, $h'(t_0)\leq0$  and $h(t_1)<0$, implies that $h(t) \leq 0$ over $(t_0,t_1)$.
	
\end{proof}

\subsection{Proof of Proposition \ref{prop5}}
\noindent
\begin{proof}
	First notice in both design problems (for Proposition \ref{prop1} and Theorem \ref{thm3}) it is optimal to set $X=1$. This is because the seller can always increases the payment to keep the ex-post payoff unchanged. Then, the payoff of a joint-accounting buyer is
	\begin{align*}
		\overline{CPT}(\theta,1,T) 
		&=  \int_{0}^{\infty} w_+(P(\theta  -T > y)) dy - \int_{-\infty}^0 w_-( P( \theta  -T < y )  ) dy \\
		&=  \int_{0}^{\theta} w(P(\theta  -y > T)) dy - \int_{0}^{\infty} w( P( T- \theta > y )  ) dy \\
		&= \int_{0}^{\theta} w(P(T<y)) dy - \int_{0}^{\infty} w( P( T > \theta+ y )  ) dy \\
		&=  \int_{0}^{\theta} (1-w(P(T\geq y))  dy - \int_{\theta}^{\infty} w( P( T >  y )  ) dy \\
		&= \theta -  \int_{0}^{\theta} w(P(T\geq y))  dy -  \int_{\theta}^{\infty} w( P( T >  y )  ) dy \\
		&= \theta - \int_{0}^{\infty} w( P( T >  y )  ) dy = CPT(\theta,1,T).
	\end{align*}
	The last equation holds because there is at most countable many $y$ such that $P(T\geq y) \neq P(T>y)$. Now because the buyer has the same payoff for any random payment $T$, all results go through.
\end{proof}

\subsection{Proof of Theorem \ref{thm2}}

\noindent
\begin{proof}
	I first prove the optimality of the proposed dynamic pricing. Consider the following program mentioned in the main text:
	\begin{gather*}
		V = \max_{\mu_k\in [0,1],F_k\in \Delta[0,1]} \quad   \sum_{k=0}^{\infty}  \frac{\mu^* \theta}{\lambda}  \delta^{k} \prod_{i=0}^{k-1}(1-\mu_i) \int  w_+(  x  )  d F_k \\
		s.t.  \quad 
		\mu_k =  \int x  d F_k, \quad \forall k.
	\end{gather*}
	Since the set of feasible $\mu_k,F_k$ is compact for any $k$, the set of feasible sequence $(\mu_k,F_k)$ is also compact due to Tychonoff's Theorem. The objective function is clearly continuous with respect to the multiple topology, thus, I know the optimal solution $(F_{k}^*)$ and optimal value $V$ exists. I then prove $V$ equals $V^*$ described in the Theorem \ref{thm2}. Note that the optimization problem of $V$ is recursive:
	\begin{gather*}
		\sum_{k=0}^{\infty}  \frac{\mu^* \theta}{\lambda}  \delta^{k} \prod_{i=0}^{k-1}(1-\mu_i) \int  w_+(  x  )  d F_k \\
		= \frac{\mu^* \theta}{\lambda} \int  w_+(  x  )  d F_0 + \delta (1-\mu_0) \Large(   	 \sum_{k=1}^{\infty}  \frac{\mu^* \theta}{\lambda}  \delta^{k} \prod_{i=1}^{k-1}(1-\mu_i) \int  w_+(  x  )  d F_k	\Large)
	\end{gather*}
	On the one hand,
	\begin{gather*}
		V = \frac{\mu^* \theta}{\lambda} \int  w_+(  x  )  d F^*_1 + \delta (1-\mu_1) V = \max_{F} \frac{\mu^* \theta}{\lambda} \int  w_+(  x  )  d F + \delta (1-\mu_1) V, \\
		\implies V = \max_F \frac{\mu^* \theta}{\lambda}   \frac{\int  w_+(  x  )  d F  }{1- \delta( 1 - \int x dF  )  }
		\geq	 \max_{x\in [0,1]} \frac{\mu^* \theta}{\lambda} \frac{w_+(x)}{1-(1-x)\delta} = V^*.
	\end{gather*}
	On the other hand, since
	\begin{gather*}
		w_+(x) - \frac{\lambda V^*}{\mu^* \theta} (1- \delta(1-x)) \leq 0, \quad \forall x \in [0,1],
	\end{gather*}
	I know
	\begin{gather*}
		\int	w_+(x) dF - \frac{\lambda V^*}{\mu^* \theta} (1- \delta(1-\int x dF)) \leq 0, \quad \forall F \in \Delta [0,1],
	\end{gather*}
	which implies $V \leq V^*$. Thus $V=V^*$. According to the analysis in the main text, I know $V=V^*$ is an upper bound of the expected revenue the seller can gain.

	The proposed dynamic pricing process $(X_t^*,T^*_t)$ will bring profit $V^*$ to the seller if the buyer keeps purchasing until she finally gets the good. Thus, the remaining task is to verify the buyer's strategy. Quitting immediate always leads to a payoff of $0$. Now consider strategy $\sigma_n$ where the buyer purchases at time $0,1,...n$ and quit thereafter (n can be $\infty$). $\sigma_0$ means purchasing in the current period and quit thereafter. By the construction of the proposed mechanism $\sigma_0$ leads to a payoff of $0$. I will show that for any $n>0$, $\sigma_n$ leads to negative payoff. 
	
	Under strategy $\sigma_n$, $\delta^{\tau} \theta$ has a discrete support: $\{ \delta^{k} \theta| k=0,1,...,n \}$ and
	\begin{gather*}
		P( \delta^{\tau} \theta > \delta^{k} \theta  )  = P( \text{Get the good before time k} ) = 1-(1-x^*)^k.
	\end{gather*}
	The gain from the good account under this strategy is:
	\begin{gather*}
		\int_0^{\infty} w_+ (  P( \delta^{\tau} \theta > y ) ) dy
		= \sum_{k=0}^{n} \Big( w_+(1-(1-x^*)^{k+1}) - w_+(1-(1-x^*)^{k}) \Big) \delta^{k} \theta.
	\end{gather*}
	Suppose by contradiction that there exists $n>0$ such that $\sigma_n$ brings non-negative payoff to the buyer. Proposition \ref{prop1} says the expected discounted payment the buyer actually pays should be strictly smaller than the upper bound:
	\begin{gather*}
		\frac{\mu^*}{\lambda } \int_0^{\infty} w_+ (  P( \delta^{\tau} \theta > y ) ) dy 
		-  \sum_{k=0}^{n}  (1-x^*)^{k}  \delta^{k} \frac{\mu^* \theta}{\lambda}  w_+(x^*)> 0.
	\end{gather*}
	Define the difference as a function of $\delta$ (normalizing $\frac{\mu^*\theta}{\lambda}=1$ to simplify notation):
	\begin{align*}
		0< D(\delta) &= \sum_{k=0}^{n} \Big( w_+(1-(1-x^*)^{k+1}) - w_+(1-(1-x^*)^{k}) \Big) \delta^{k}  -  \sum_{k=0}^{n}  (1-x^*)^{k}    w_+(x^*) \delta^{k}\\
		&= \sum_{k=0}^{n} \Big( w_+(1-(1-x^*)^{k+1}) - w_+(1-(1-x^*)^{k}) - (1-x^*)^{k}   w_+(x^*) \Big) \delta^{k}\\
		&= \sum_{k=0}^{n} a_k \delta^{k}.
	\end{align*}
	Now let me analysis the sign of $a_k$:
	\begin{align*}
		\frac{a_k}{(1-x^*)^{k} x^*} &= \frac{ w_+(1-(1-x^*)^{k+1}) - w_+(1-(1-x^*)^{k}) - (1-x^*)^{k}   w_+(x^*)}{(1-x^*)^{k} x^*}  \\
		&=  \frac{ w_+(1-(1-x^*)^{k+1}) - w_+(1-(1-x^*)^{k}) }{1-(1-x^*)^{k+1}-1+(1-x^*)^{k}} - \frac{w_+(x^*)}{x^*}.
	\end{align*}
	By mean value theorem, there exists $x_k \in [1-(1-x^*)^{k},1-(1-x^*)^{k+1}  ]$ such that
	\begin{gather*}
		\frac{ w_+(1-(1-x^*)^{k+1}) - w_+(1-(1-x^*)^{k}) }{1-(1-x^*)^{k+1}-1+(1-x^*)^{k}}= w_+'(x_k). 
	\end{gather*}
	Because $w_+$ is first concave and then convex, $a_k$ is first decreasing then increasing. Additionally, note that $a_0 = w_+(x^*) - 0 - w_+(x^*) = 0$. Therefore, I know there exists a number $k_0<n$ (if $k_0 \geq n$ then $D(\delta) \leq  0$, a contradiction) such that:
	\begin{gather*}
		a_k \leq 0, \text{ if }k \leq k_0; \quad a_k > 0 \text{ if }k> k_0.
	\end{gather*}
	If $n<\infty$, I can calculate $D'(\mu)$ for $\mu \in [\delta,1]$ by taking derivative by terms:
	\begin{align*}
		D'(\mu) &= \sum_k a_k  k   \mu^{k-1}  =
		\sum_{k\leq k_0}a_k  k   \mu^{k-1}  + \sum_{k> k_0}a_k  k   \mu^{k-1} \nonumber\\
		&\geq  \sum_{k\leq k_0}a_k  k_0   \mu^{k-1}  + \sum_{k> k_0}a_k  k_0   \mu^{k-1} \\
		&= \sum_{k}a_k  k_0   \mu^{k-1}   = k_0 D(\mu) > 0. \nonumber
	\end{align*}
	Consequently $D(1) >  D(\delta) >0$. The case of $n = \infty$ can be done by taking the $n$ to the limit.  In general the weak inequality always holds: 
	\begin{gather*}
		D(1) = w_+(1-(1-x^*)^{n+1}) - \frac{w_+(x^*)}{x^*}(1-(1-x^*)^{n+1})  \geq 0.\\
		\Rightarrow \frac{w_+(x^*)}{x^*} \leq \frac{w_+(1-(1-x^*)^{n+1})}{1-(1-x^*)^{n+1}} .
	\end{gather*}
	When $n>0$, $p_n = 1-(1-x^*)^{n+1} > x^*$. This jointly with the above inequality implies
	\begin{gather*}
		\frac{\frac{w_+(x^*)}{x^*}}{\frac{1}{x^*}(1-\delta) + \delta} < \frac{\frac{w_+(p_n)}{p_n}}{\frac{1}{p_n}(1-\delta) +\delta}.
	\end{gather*}
	However, because
	\begin{gather*}
		x^*\in \argmax  \frac{w_+(x)}{1-(1-x)\delta},\\
		\Rightarrow \frac{w_+(x^*)}{1-(1-x^*)\delta} \geq \frac{w_+(p_n)}{1-(1-p_n)\delta},\\
		\Rightarrow  \frac{\frac{w_+(x^*)}{x^*}}{\frac{1}{x^*}(1-\delta) + \delta} \geq  \frac{\frac{w_+(p_n)}{p_n}}{\frac{1}{p_n}(1-\delta) +\delta} .
	\end{gather*}
	This is a contradiction.
	
	In conclusion, for any $n>0$, $\sigma_n$ gives the buyer negative payoff. Thus, in the seller's favorite equilibrium, the buyer always chooses $\sigma_0$ in each period. Namely, she purchases at the current period, naively thinking that she will quit thereafter. The proposed mechanism achieves the upper bound of the profit $V^*$ and is therefore the optimal one.
	
	\paragraph{Uniqueness}
	The uniqueness part of the theorem is proved by carefully analyzing binding constraints of all relaxations made in the main text. 
	Now I prove the essential uniqueness of the optimal pricing process. I first show that for any fixed weighting function $w_+$, except for at most one $\delta$, the maximizer of the problem
	\begin{gather*}
		y \equiv  \frac{\lambda V^*}{\mu^* \theta} =  \max_F \frac{\int  w_+(  x  )  d F  }{1- \delta( 1 - \int x dF  )  }
	\end{gather*}
	is unique.
	Denote $g(x)$ as the smallest concave function that is greater than $w_+(x)$. Since $w_+(x)$ is strictly concave in $[0,p_+]$ and strictly convex in $[p^+,1]$. We know:
	\begin{gather*}
		g(x) = \begin{cases}
			w_+(x)  & \text{ if } x \in [0,p^+],\\
			\frac{ 1-x}{1-p^{+}} w_+(p^+) + \frac{x-p^+}{1-p^+} w_+(1) & \text{ if } x \in [p^+,1].
		\end{cases}
	\end{gather*}
	By the definition of $y$,
	\begin{align*}
		w_+(  x  )   &\leq y  (1- \delta( 1 -  x   )), \quad \forall x \in [0,1], \\
		w_+(  x_0  )   &= y  (1- \delta( 1 -  x_0   )), \text{ for some } x_0 \in [0,1],
	\end{align*}
	Notice $y  (1- \delta( 1 -  x   ))$ is concave in $x$, by the definition of $g$,
	\begin{align*}
		g(  x  )   &\leq y  (1- \delta( 1 -  x   )), \quad \forall x \in [0,1], \\
		w_+(  x_0  )   &= g(x_0) = y  (1- \delta( 1 -  x_0   )).
	\end{align*}
	Thus $y  (1- \delta( 1 -  x   ))$ is a supporting plane of $g$ at $x_0$. Because $g$ is strictly concave in $[0,p^+]$ and linear in $[p^+,1]$, $x_0 \in  [p^+,1]$ if and only if the slope of $y  (1- \delta( 1 -  x   ))$ coincides with $g$:
	\begin{gather*}
		\delta y = \frac{w_+(1)-w_+(p^+)}{1-p^+}
	\end{gather*} 
	It is clear that $y$ as a function of $\delta$ is strictly increasing in $\delta$, thus there is at most one $\delta$ satisfying this equation. For all other $\delta$, $x_0 \in [0,p^+]$. Because $g$ is strictly concave in $[0,p^+]$, we know $ y  (1- \delta( 1 -  x   )) > g(x)\geq w_+(x)$ for all $x \neq x_0$. Therefore, the degenerated distribution at $x_0$ is the unique maximizer:
	\begin{gather*}
		\delta_{x_0} = \argmax_F \frac{\int  w_+(  x  )  d F  }{1- \delta( 1 - \int x dF  )  }.
	\end{gather*}
	
	From now I focus on $\delta$ where $x_0=x^*$ is the unique maximizer. Consider any mechanism $(X_t,T_t)$ and the buyer's strategy $(\sigma^t)_t$ that achieve the optimal profit $V^*$. Because $\sigma^t$ satisfies SIR conditions under $(X_t,T_t)$,
	I can repeat the same relaxation as I did in the main text. Because $(X_t,T_t)$ and $(\sigma^t)_t$ indeed achieves the upper bound, all the relaxation must be tight. This implies:
	\begin{align}
		\hat{\tau}_k &= k, \quad \tau_k =k+1,  \label{eq7} \\
		\mu_k &= P ( \hat{\tau}_k \leq \tau < \tau_k  |  h_{\hat{\tau}_k} ,\tau\geq \hat{\tau}_k)  =  x^*,  \label{eq8}  \\
		\mathbb{E} [ \sum_{t= \hat{\tau}_k}^{\tau_k-1}  \delta^{t-\hat{\tau}_k}     T_t | h_{\hat{\tau}_k},\tau\geq \hat{\tau}_k ] &=  \frac{\mu^* \theta}{\lambda}  w_+( P ( \hat{\tau}_k \leq \tau < \tau_k |  h_{\hat{\tau}_k},\tau\geq \hat{\tau}_k ) ) \label{eq9}.
	\end{align}
	Equation \ref{eq7} means the buyer actually purchases in every period (conditioned on not getting the good), so on path history is $H^*$. This jointly with equation \ref{eq8} implies 
	\begin{gather*}
		P(X_t = 1|h_t \in H^*) = x^*.
	\end{gather*}
	These jointly with equation \ref{eq9} imply
	\begin{gather*}
		\mathbb{E} [    T_{t}| h_t \in H^* ] = \frac{\mu^* \theta}{\lambda}   w_+ (   x^*  ).
	\end{gather*}
	According to Proposition \ref{prop1}, the unique $T_{t}$ that has this expected revenue and satisfies SIR condition is the binary random price described in Theorem \ref{thm2}. This completes the proof.
\end{proof}

\bigskip

\subsection{Proof for Section 6.2}
\begin{proof}[Proof of Proposition \ref{prop_salience}]
	Because it is optimal to set $X=1$ for sure, the design problem is to design the distribution of $T$. I abuse notations to let $F$ denote the distribution of $T$. The problem is
	\begin{gather*}
		\Pi^* = \max_{F} \int y d F(y) \\
		\int (\theta  - y ) \sigma(\theta  - y,0)  dF(y) \geq 0.
	\end{gather*}
	Denote the Distribution of $\theta-y$ as $G$, the problem is equivalent to
	\begin{gather*}
		\theta - \Pi^* = \min_{G\in \Delta(-\infty,\theta]  } \int x d G(x) \\
		s.t. \quad \int x \sigma(x,0)  dG(x) \geq 0.
	\end{gather*}
	Or equivalently
	\begin{gather*}
		\theta - \Pi^* = \min_{\pi} \pi \\
		s.t. \quad 0 \leq  \max_{G\in \Delta(-\infty,\theta]} \int x \sigma(x,0) dG(x) \\
		s.t.\int x dG(x) = \pi.
	\end{gather*}
	The inner problem is:
	\begin{gather*}
		V(\pi) = \max_{G\in \Delta(-\infty,\theta]} \int x \sigma(x,0) dG(x) \\
		\int x dG(x) = \pi.
	\end{gather*}
	It  can be solved by standard convexification. 
	
	Let us first consider $V(0)$. When $\pi=0$, because $x \sigma(x,0)$ is strictly convex when $x>0$ and strictly concave when $x<0$, and the left derivative of $x \sigma(x,0)$ at $x=0$ is weakly smaller than the right derivative, the unique solution of convexification is to convexify $\theta$ and some point $x<0$. This also means $V(0)>0$ and so $\Pi^* < 0$. One can also see that if the price $T$ can be arbitrarily negative so $G \in \Delta \mathbb{R}$, the optimal value $V(0)$ does not exist.
	
	For $\pi<0$,  because $x \sigma(x,0)$ is strictly convex when $x>0$ and strictly concave when $x<0$, the optimal convexification that achieves $V(\pi)$ is unique: it is either singleton or a binary distribution over $\theta$ and some $x<0$. Whenever the optimal solution of $V(\pi)$ is a singleton, $V(\pi) = \pi \sigma(\pi,0) < 0$. Because $V(\pi)$ is clearly continuous in $\pi$, 
	\begin{gather*}
		\theta-\Pi^* = \min \{ \pi | V(\pi) \geq 0  \}.
	\end{gather*}
	and the unique optimal convexification of $V(\theta-\Pi^*)$, which is supported on $\theta$ and some negative price, is the unique binary random price in the optimal mechanism.
\end{proof}

\bigskip 
\subsection{Proof for Section 6.3}
\noindent
\begin{proof}[Proof of Lemma \ref{lemm3}]
	By the construction of $P^*$,
	\begin{align*}
		\lambda \int_0^{\infty} w_-( P^*_{\theta}(y)  ) dy &= \theta w_+ (x(\theta) ) - v(\theta)\\
		&=\lambda \int_0^{\infty} w_-( P_{\theta}(y)  ) dy.
	\end{align*}
	Note that the interim payoff for type $\theta$ who reports $\hat{\theta}$ in the original mechanism is:
	\begin{gather*}
		v(\theta,\hat{\theta}) = \theta w_+(x(\theta)) - 	\lambda \int_0^{\infty} w_-( P_{\hat{\theta}}(y)  ) dy.
	\end{gather*}
	Therefore under two direct mechanisms $(x,P)$ and $(x,P^*)$, the interim payoff is the same. Thus $(x,P^*)$ automatically satisfies IR and IC.
	
	On the other hand,
	according to Proposition \ref{prop1}, for any $\theta$,
	\begin{gather*}
		\int_0^{\infty} P(T^*(\theta) > y) dy \geq \int_0^{\infty} P(T(\theta) > y) dy,
	\end{gather*}
	and the inequality is strict if $T^*(\theta) \overset{d}{\neq} T(\theta)$.  This completes the proof.
\end{proof}	

\bigskip

\noindent
\begin{proof}[Proof of Proposition \ref{prop4}]
	Recall that the design problem is:
	\begin{gather*}
		\max_{x,P}   \int_{\underline{\theta}}^{\overline{\theta}}  \int_0^{\infty} P(y) dy   ~ f(\theta) d\theta , \\
		s.t. \qquad
		CPT(\theta,x(\theta),P_{\theta}) \geq  CPT(\theta,x(\hat{\theta}),P_{\hat{\theta}}), \\
		v(\theta) = CPT(\theta,x(\theta),P_{\theta}) \geq 0.
	\end{gather*}
	
	Following the standard approach of \cite{Myerson81}, recall that the equilibrium value function $v(\theta)$ of buyer for  is defined as
	\begin{gather*}
		v(\theta) =    \theta w_+ (x(\theta) )  ) - \lambda \int_0^{\infty} w_-( P_{\theta}(y)  ) dy
	\end{gather*}
	It is straightforward to duplicate the envelop theorem, whose proof I skip.
	\begin{lemma}
		\label{lemma2}
		For any incentive compatible mechanism $(x,P)$, $v(\theta)$ is continuous and convex (and therefore absolute continuous). The derivative $v'(t)$ exists almost every where and in addition,
		\begin{gather*}
			v(\theta) = v(\underline{\theta}) +  \int_{\underline{\theta}}^{\theta}  w_+ ( x(s) ) ds.
		\end{gather*}
	\end{lemma}
	
	According to Lemma \ref{lemm3}, for any allocation rule $x(\theta)$ it is without loss of generality to focus on the optimal binary price
	\begin{gather*}
		P^*_{\theta}(y)  = \begin{cases}
			0   \quad  &\text{ if }y \geq \frac{\mu^*}{\lambda p^*}  \Big( \theta w_+ (x(\theta) ) - v(\theta) \Big),\\
			p^* &\text{ if }y <   \frac{\mu^*}{\lambda p^*}  \Big( \theta w_+ (x(\theta) ) - v(\theta) \Big) .
		\end{cases}
	\end{gather*}	
	Combining these two observation, the design problem reduces to
	\begin{gather*}
		\max_{x(\theta),~v(\theta)}   \int_{\Theta} \Large(  \frac{ \mu^*}{\lambda}( \theta w_+ (x(\theta)  ) - v(\theta))   \Large)  ~ f(\theta) d\theta , \\
		s.t. \\
		v(\theta) = v(0) + \int_0^{\theta} w_+ (x(s)  )  ds,  \\
		w_+ (x(\theta) ) \text{ is increasing in } \theta, \\
		v(0) \geq 0.
	\end{gather*}
	Clearly, it is optimal to set $v(0)=0$. Integrating by part I get:
	\begin{gather*}
		\max_{q(\theta)}   \int_{\Theta}   \frac{ \mu^*}{\lambda}( \theta  - \frac{1-F(\theta)}{f(\theta)} )w_{+}(\theta)   ~ f(\theta) d\theta , \\
		s.t. \quad 
		w_{+}(\theta)  \text{ is increasing in } \theta. 
	\end{gather*}
	When the common prior $F$ is regular so that the virtual value function $\phi(\theta) = \theta - (1-F(\theta))/f(\theta) $ is increasing in $\theta$, simple point-wise maximization leads to the result.  When $F$ is not regular, one can use the standard ironing technique and replace the virtual value with the ironed version $\bar{\phi}(\theta)$.
\end{proof}

\end{document}